\documentclass[a4paper,11pt]{article}
\usepackage{jheppub} 
\usepackage{hyperref}
\usepackage{graphicx}
\usepackage[T1]{fontenc}
\usepackage[utf8]{inputenc}
\usepackage[english]{babel}
\usepackage[labelfont=bf, font=small]{caption}
\usepackage{caption}
\usepackage{subcaption}
\usepackage{physics}
\usepackage{titlesec}
\usepackage{setspace}
\usepackage{slashed}
\usepackage{amsmath,amsthm,amsfonts,amssymb,amsopn,amscd}

\newcommand{\eqdef}{\overset{\mathrm{def}}{=}}
\newtheorem{Definition}{Definition}[subsection]
\newtheorem{Proposition}{Proposition}[subsection]

\title{\boldmath A QFT information protocol for charged black holes}

\author{Paolo Palumbo}
\affiliation{Dipartimento di Fisica "Ettore Pancini",
Universit\`a degli studi di Napoli Federico II,
Via Cintia, 80126 Napoli, Italy}
\affiliation{INFN, Sezione di Napoli, Complesso Universitario di Monte S. Angelo, Via Cintia Edificio 6, 80126 Napoli, Italy}

\emailAdd{paolo.palumbo@unina.it}

\abstract{A generalization for the quantum information retrieval protocol recently illustrated by Verlinde and van der Heijden for evaporating black holes is provided to inclusions of type $\mathrm{III}$ von Neumann factors. The physical interest of such scenario arises in Quantum Field Theory, where local algebras are type $\mathrm{III}$ von Neumann algebras. The formula obtained can be easily interpreted in terms of the statistical dimension of superselection sectors in the case of black holes undergoing charge evaporation, thanks to the index-statistics theorem, leading to a thermodynamic interpretation. A constraint on the values of the index leads to a final remark about the quantization of the charge emitted by the black hole during the evaporation process.}

\begin{document}
\maketitle
\flushbottom

\section{Introduction}
In ordinary quantum mechanics, we know that given two subsystems $a$ and $b$, we can associate a Hilbert space $\mathcal{H}_a$ and $\mathcal{H}_b$ to each and the global system will be described in terms of the tensor product $\mathcal{H}=\mathcal{H}_a \otimes \mathcal{H}_b$. Such tensor product structure is responsible for some quantum aspects of nature, like the existence of entangled states.

In the relativistic scenario, observables become \textit{local}, because of relativistic covariance: a relativistic quantum mechanics is necessarily a Quantum Field Theory (QFT). Consequently, we can define a map which associates a local algebra of observables $\mathcal{A}(\mathcal{O})$ to every open subset of spacetime $\mathcal{O}$\footnote{A concrete realization is given by the algebra generated by fields smeared on $\mathcal{O}$, i.e. with a suitable smearing function having compact support in $\mathcal{O}$.} \cite{Haag} \cite{Haag-Kastler}. Causally disconnected spacetime regions will result in commuting algebras. The striking feature of this formalism is that if one considers the \textit{subsystem} associated to a bounded spacetime region $\mathcal{O}$, the Hilbert space of the global spacetime, namely the Fock space of the theory, will not factorize as in non-relativistic quantum mechanics\footnote{This is true referring to the spacetime degrees of freedom, not to the field modes which are delocalized.}. For instance, if one were to consider the Cauchy slice $t=0$ in Minkowski spacetime, its associated Hilbert space does not factorize into a tensor product $\mathcal{H}_L \otimes \mathcal{H}_R$, with these Hilbert spaces being associated to the subregions $x<0$ and $x>0$ respectively. More specifically, Hilbert spaces associated to the left and right Rindler wedges do not exist at all. As a consequence, one cannot define reduced density matrices for subregions of spacetime: since all quantum information algorithms hold onto the Hilbert factorization assumption, one cannot extend them trivially to the relativistic regime \cite{Werner}. Mathematically, the reason comes from the observation that local algebras of QFT are type $\mathrm{III}$ von Neumann algebras\footnote{For the free scalar field the algebra of a Rindler wedge is a type $\mathrm{III}_1$ factor \cite{Connes1}.} \cite{Araki}. Operator algebras provide a natural framework to deal with the generalization of quantum information tasks in the relativistic regime.
On that note, in \cite{Verlinde} a generalization of the information retrieval protocol for black holes was achieved in the case of type $\mathrm{II}_1$ algebras thanks to the theory of Jones construction. The idea to generalize quantum teleportation in algebraic terms has already initiated in \cite{Conlon}. In this work, we discuss a possible extension of such protocol to the type $\mathrm{III}$ case, which better matches the framework of QFT. However, since spacetime inclusions cannot have finite index, the protocol cannot be applied generally to spacetime inclusions.

The paper is structured as follows.
In the first part, we briefly describe the quantum teleportation protocol for qubits both in the Schrödinger picture and the Heisenberg picture, the latter better capturing the information transfer in terms of observables.
Then, we recall the idea developed by Verlinde and van der Heijden for the algebraic picture of the protocol, emphasising the roles covered by Jones index and Jones projections.
Afterwards, we mention the mathematical definition and main properties of Kosaki-Longo index, which provides a generalization of Jones index in the type $\mathrm{III}$ case.
In the final part, we present a version of the protocol for type $\mathrm{III}$ factors, also covering the case of charged sectors in QFT with its application to black hole physics.

\section{Quantum Teleportation}
Entanglement \cite{Schrodinger} \cite{EPR} has long been exploited to perform quantum information protocols in physics. The first task often encountered in graduate courses involves the teleportation of a qubit from a first observer, Alice, to another observer, Bob \cite{Teleportation}. We assume that Alice and Bob are "far away" from each other, in the sense that they have mutually commuting observables at their disposal. 
This provides a first example aimed at highlighting the core differences between classical and quantum physics and makes explicit use of the phenomenon of entanglement, which is one of the main features present only in the quantum realm.

\subsection{Quantum teleportation for qubits}

Consider a tripartite system of qubits, two of each are at Alice's disposal and the remaining one belongs to Bob. The system is described by the algebra $M_2(\mathbb{C}) \otimes M_2(\mathbb{C}) \otimes M_2(\mathbb{C})$. Supposing the state of the system is prepared such that Bob's qubit is entangled with one of Alice's, we have:

\begin{equation}
    \ket{\Psi} = \ket{\psi} \otimes \ket{\beta_{11}}
\end{equation}
with

\begin{equation}
    \begin{split}
        &\ket{\beta_{11}} \eqdef \frac{1}{\sqrt{2}} (\ket{00} + \ket{11}) \\
        &\ket{\beta_{12}} \eqdef \frac{1}{\sqrt{2}} (\ket{00} - \ket{11}) = (Z \otimes I) \ket{\beta_{11}} \\
        &\ket{\beta_{21}} \eqdef \frac{1}{\sqrt{2}} (\ket{10} + \ket{01}) = (X \otimes I) \ket{\beta_{11}} \\
        &\ket{\beta_{22}} \eqdef \frac{1}{\sqrt{2}} (\ket{01} - \ket{10}) = (ZX \otimes I) \ket{\beta_{11}}
    \end{split}
\end{equation}
representing the Bell basis of $\mathbb{C}^2 \otimes \mathbb{C}^2$. $X$ and $Z$ correspond to the first and third Pauli matrices. The task we want to perform is to apply local operations to the system to have Bob's qubit in the state $\ket{\psi}$, which results in teleporting Alice's information from the first qubit onto Bob's one.

We can rewrite $\ket{\Psi}$ as

\begin{equation}
    \begin{split}
        \ket{\Psi} &= \ket{\psi} \otimes \ket{\beta_{11}} \\
        &= \frac{1}{2} (\ket{\beta_{11}} \otimes\ket{\psi} + \ket{\beta_{12}} \otimes Z \ket{\psi} + \ket{\beta_{21}} \otimes X \ket{\psi} + \ket{\beta_{22}} \otimes XZ \ket{\psi}) \\
        &= \sum_{i, j = 1}^2 \ket{\beta_{ij}} \otimes (X^{i-1}Z^{j-1})\ket{\psi} \text{.}
    \end{split}
    \label{initial state}
\end{equation}
Alice then checks whether her part of the system is in the state $\ket{\beta_{ij}}$; this measurement will yield the output $(i,j)$. After that, she classically sends the result to Bob, who correspondingly applies the unitary operation $Z^{j-1}X^{i-1}$. The resulting state is

\begin{equation}
    \ket{\Phi_{ij}} = \ket{\beta_{ij}} \otimes \ket{\psi}
\end{equation}
hence the teleportation of the state $\ket{\psi}$ from Alice to Bob worked.

\subsection{Heisenberg picture of the teleportation scheme}

We shall now reformulate the teleportation protocol in the Heisenberg picture of quantum mechanics \cite{Bedard} \cite{Conlon}, to better elucidate the switch to the algebraic approach discussed later.

The initial state of the system can be expressed in terms of the density matrix $\rho \otimes \omega$, where $\rho = \ket{\psi} \bra{\psi}$ and $\omega = \ket{\beta_{11}} \bra{\beta_{11}}$. The decomposition \ref{initial state} translates into

\begin{equation}
    \rho \otimes \omega = \sum_{i,j=1}^2 (U_{ij}^* \otimes I) \omega (U_{ij} \otimes I) \otimes U_{ij}^* \rho U_{ij}
\end{equation}
with $U_{ij} \eqdef X^{i-1} Z^{j-1}$. Alice's measurement consists in the application of the operator $M_{ij} \eqdef (U_{ij}^* \otimes I) \omega (U_{ij} \otimes I) = \ket{\beta_{ij}} \bra{\beta_{ij}}$ for the result $(i,j)$, where $I$ is the identity: the resulting state becomes

\begin{equation}
    \sum_{i,j=1}^2 (M_{ij} \otimes I)(\rho \otimes \omega)(M_{ij} \otimes I) \text{;}
\end{equation}
Bob's corresponding operation leads to the application of $U_{ij}^*$ on his part of the system:

\begin{equation}
    \sum_{i,j=1}^2 (M_{ij} \otimes U_{ij}^*)(\rho \otimes \omega)(M_{ij} \otimes U_{ij}) \text{.}
\end{equation}
Finally, Bob can apply a projection $P \in M_2(\mathbb{C})$ and reproduce the information contained in Alice's first qubit: the probabilities to have \textit{yes}-answers for any $P$ in Alice's possession concerning the first qubit are unaffected and correctly reproduced by Bob on his qubit after the previous operations. Indeed,

\begin{equation}
    \mathrm{Tr}(\rho P) = \sum_{i,j=1}^2 \mathrm{Tr} \big( (M_{ij} \otimes P U_{ij}^*)(\rho \otimes \omega)(M_{ij} \otimes U_{ij} P) \big) \text{.}
    \label{Success}
\end{equation}
Since the algebras considered are von Neumann algebras, they are spanned by their projections and hence the previous condition extends for any $A \in M_2(\mathbb{C})$. This provides an example of LOCC (\textit{Local Operations and Classical Communication}) protocol.

\section{Jones Index and the protocol}

 In \cite{Verlinde}, an operator-algebraic generalization of the Hayden-Preskill \cite{Hayden} protocol setup is presented. The work translates the task into algebraic language and provides a new physical interpretation for the Jones index \cite{Jones}, a quantity defined in the context of inclusions of type $\mathrm{II}_1$ factors\footnote{Factors are von Neumann algebras with trivial center. See \ref{Factors}.}. The result relates such quantity to a map between the algebras of two causally disconnected observers, allowing quantum teleportation between the two. 

In this part we sum up their results. With arguments from the AdS/CFT correspondence, the authors claim that the algebra arising in a certain limit is a type $\mathrm{II}_1$ factor. The emergence of a non-type-$\mathrm{III}$ algebra in this limit suggests that we are not describing the QFT-limit, but rather a sort of quantum geometry inside the black hole, including the treatment of the \textit{singularity} (which may not be such in the quantum gravity regime). Type $\mathrm{II}$ algebras do appear often when gravitational degrees of freedom are taken into account \cite{Witten-Longo}. In the type $\mathrm{II}_1$ setup, all the algebras admit traces defined on every element, which is not the case for type $\mathrm{III}$ factors. In this paper, the main contribution will be describing the strict QFT limit ($G_N \rightarrow 0$), in which type $\mathrm{III}$ algebras are supposed to emerge.

\subsection{The algebras of Alice and Bob}

We move on to describing the roles of the observers in the protocol and their corresponding algebras of observables. Suppose that an early observer, Alice, has stored some information in a diary and thrown it into an evaporating black hole \cite{Hawking}. We shall assume that the diary is small, in the sense that Alice's information can be retrieved by a small amount of radiation energy $E$. We denote with $\mathcal{M}$ the algebra of the black hole of mass $M$ and with $\mathcal{N}$ the algebra of the black hole of mass $M - E$, i.e. after an evaporation step which radiates the information content of Alice's diary. $\mathcal{N}$ is a subalgebra of $\mathcal{M}$, since the black hole shrinks after the evaporation step (see figure \ref{Black hole algebras}). The fact that Alice's information is encoded in the radiation of energy $E$ can be formalized as stating that Alice's algebra is

\begin{equation}
    \mathcal{A} \eqdef \mathcal{N}' \cap \mathcal{M} \text{,}
\end{equation}
as the shrunk black hole has lost all information about Alice's diary.

Let us switch to Bob's point of view: he is another observer who has access to the early radiation and, at first, has no access to the black hole algebra $\mathcal{M}$. However, after the evaporation step described above,  Bob will receive the outgoing Hawking radiation, although it does not directly correspond to the information of Alice's diary: he cannot simply read off such information from the radiation. Let us introduce the following inclusions:

\begin{equation}
    \mathcal{M} \subset \mathcal{M}_1 \subset \mathcal{M}_2 \text{,}
\end{equation}
where $\mathcal{M}_1$ contains $\mathcal{M}$ and additionally the algebra of the infalling Hawking radiation, while $\mathcal{M}_2$ contains both the algebra of the infalling and outgoing radiation. Notice that the infalling and outgoing Hawking pairs are entangled and such entanglement is what allows the teleportation protocol to work. Bob's algebra is then defined as

\begin{equation}
    \mathcal{B} \eqdef \mathcal{M}_1' \cap \mathcal{M}_2 \text{.}
\end{equation}
In fact, Bob cannot access the information of the infalling Hawking particles, since it lays beyond the horizon, but can probe the outgoing radiation.
\begin{figure}[h!]
    \centering
    \includegraphics[width=0.5\textwidth]{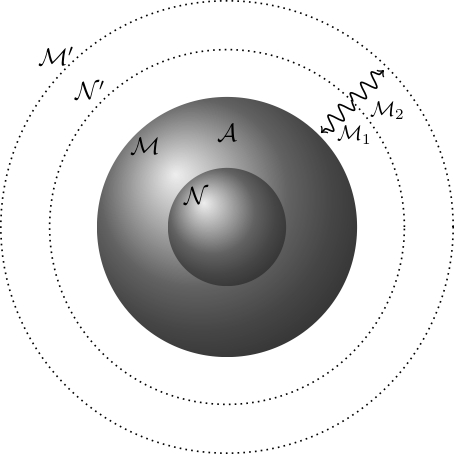}
    \caption{Visualization of the algebras introduced above. The picture is to be taken as a reference and does not represent the true picture.}
    \label{Black hole algebras}
\end{figure}
Our claim is that Alice's information can be recovered from the radiation: concretely, Bob can reproduce Alice's correlation functions. The inclusions of algebras sketched above are realized mathematically through Jones basic construction and Jones index theory, which has been known in literature for about half a century for the case of type $\mathrm{II}_1$ factors.

From Tomita-Takesaki theorem \cite{Takesaki} \cite{Takesaki2}, we have

\begin{equation}
    \mathcal{M}_1' = J_{\mathcal{M}_1} \mathcal{M}_1 J_{\mathcal{M}_1} \text{,}
    \label{M1 and M1'}
\end{equation}
where the modular conjugation can be constructed from a cyclic separating vector. In the type $\mathrm{II}_1$ case, such vector could be the tracial state $\ket{\Psi}$; however, we have to deal with multiple GNS representations as the tracial state vector changes when changing representations. This is not the case for type $\mathrm{III}$ factors, where we can choose a unique representation for all algebras, for instance the vacuum one in QFT.

There is a map, called \textit{canonical shift}, which naturally connects Alice's and Bob's algebras:

\begin{equation}
    \mathcal{B} = \Gamma (\mathcal{A}) \eqdef J_{\mathcal{M}_1} J_{\mathcal{M}} \mathcal{A} J_{\mathcal{M}} J_{\mathcal{M}_1} \text{.}
\end{equation}
In the type $\mathrm{II}_1$ case discussed in \cite{Verlinde}, the KMS condition implies the triviality of the modular operator associated with the tracial state; in the type $\mathrm{III}$ case that is generally not true, although in the case of black holes we know the form of the modular operator, at least in simple cases, due to the equivalence principle and Bisognano-Wichmann theorem \cite{Bisognano-Wichmann} \cite{Sewell}.

\subsection{Conditional expectations and Jones basic construction}

For Bob to be able to recover the information contained in Alice's diary, he needs to be able to \textit{trace out} Alice's information from the black hole: this operation, in the case of generic von Neumann algebras, is realized through the \textit{conditional expectation}

\begin{equation}
    \mathcal{E} : \mathcal{M} \rightarrow \mathcal{N}
\end{equation}
defined to have the properties:

\begin{itemize}
    \item $\mathcal{E}(I) = 1$;
    \item $\big( \mathcal{E}(M) \big)^* = \mathcal{E}(M^*), \ \forall M \in \mathcal{M}$;
    \item $\mathcal{E}(N_1 M N_2) = N_1 \mathcal{E}(M) N_2, \ \forall M \in \mathcal{M}, \ \forall N_1, N_2 \in \mathcal{N}$.
\end{itemize}
The last one is called the \textit{bimodule property}. The \textit{Jones projection} is defined in terms of the conditional expectation as

\begin{equation}
    e_{\mathcal{N}} M e_{\mathcal{N}} \eqdef \mathcal{E}(M) e_{\mathcal{N}}
\end{equation}
and satisfies the usual properties of projections. Let us introduce now the algebra

\begin{equation}
    \mathcal{M}_1 \eqdef \langle \mathcal{M}, e_{\mathcal{N}} \rangle \text{,}
\end{equation}
where $\langle \cdot, \cdot \rangle$ denotes the smallest von Neumann algebra generated by its two arguments, namely $( \mathcal{M} \cup \{ e_{\mathcal{N}} \} )''$ in our case. This algebra corresponds to the algebra of the black hole and the infalling particle: in fact, we can express it as

\begin{equation}
    \mathcal{M}_1 = J_{\mathcal{M}} \mathcal{N}' J_{\mathcal{M}} \text{.}
    \label{M1 and N'}
\end{equation}
Notice that $\mathcal{N}' \supset \mathcal{M}'$, so $\mathcal{M}_1 = J_{\mathcal{M}} \mathcal{N}' J_{\mathcal{M}} \supset J_{\mathcal{M}} \mathcal{M}' J_{\mathcal{M}} = \mathcal{M}$, as it should be. If relation \ref{M1 and N'} holds, we can interpret the action of $J_{\mathcal{M}}$ as taking the particle conjugate of the outgoing radiation, which is indeed the infalling one (again, see figure \ref{Black hole algebras}).

\subsection{The information retrieval protocol}

Suppose the act of Alice creating her diary can be modelled by a suitable operator $X \in \mathcal{A}$ acting on a state $\ket{\Psi}$: $X \ket{\Psi} = \ket{X}$. For the type $\mathrm{II}$ case discussed in \cite{Verlinde}, we may pick $\ket{\Psi}$ as the tracial state vector. Bob's goal is to reproduce all Alice's correlation functions $\bra{X} A \ket{X}$ (for any $A \in \mathcal{A}$), but he does not have direct access to Alice's algebra. 

If we construct the algebra

\begin{equation}
    \mathcal{M}_2 \eqdef \langle \mathcal{M}_1, e_{\mathcal{M}} \rangle \text{,}
\end{equation}
the canonical shift $\Gamma$ maps the relative commutant $\mathcal{A}$ into the relative commutant $\mathcal{B}$. Bob can reproduce Alice's correlation functions using his algebra:

\begin{equation}
    \begin{split}
        &\bra{e_{\mathcal{N}} X} \Gamma (A) \ket{e_{\mathcal{N}} X} = \bra{X} e_{\mathcal{N}} A \ket{X} = \mathrm{Tr}_{\mathcal{M}_1} (e_{\mathcal{N}} A X X^*) = \\
        &[\mathcal{M} : \mathcal{N}]^{-1} \mathrm{Tr}_{\mathcal{M}} (A X X^*) \text{,}
    \end{split}
    \label{fResult}
\end{equation}
having used the so-called \textit{Markov property}

\begin{equation}
    \mathrm{Tr}_{\mathcal{M}_1}(e_{\mathcal{N}} A) = [\mathcal{M} : \mathcal{N}]^{-1} \mathrm{Tr}_{\mathcal{M}}(A), \ \forall A \in \mathcal{A} \text{.}
    \label{Jind}
\end{equation}
The quantity $[\mathcal{M} : \mathcal{N}]$ in the previous formula is called the \textit{Jones index} of $\mathcal{N}$ in $\mathcal{M}$. In the appendix, there is a section dedicated to its definition in the type $\mathrm{II}_1$ case and the proof of equation \ref{Jind}; we will discuss its generalization to the type $\mathrm{III}$ in the following section.

Finally, Bob can reproduce all of Alice's correlation functions using $\Gamma (A) \in \mathcal{B}$, after a suitable normalization. This relates to equation \ref{Success} and corresponds to a successful teleportation protocol. Since the canonical shift satisfies the algebraic structure of operators, we have that Bob can indeed reproduce the information of Alice's diary from the radiation.

In this discussion, the dynamics of the radiation-black hole system was ignored: during the evaporation step, there actually is a unitary operator $U$ describing time evolution. If Bob knew the form of $U$ and had such operator at his disposal, he could apply $U^*$ to his state to undo the evolution; however, that is not physically the case. Hayden and Preskill \cite{Hayden} showed that, under certain assumptions, Bob could infer the the form of $U$ and undo the dynamics, reproducing Alice's correlation functions. Further progress has been made on the matter, as $U$ does not generally belong to Bob's algebra: it can be shown that there are suitable operators at Bob's disposal which almost undo the dynamical evolution, allowing for the approximate retrieval of Alice's correlation functions, at least in the type $\mathrm{I}$ case \cite{Yoshida}.

\section{Jones Index for type III factors}
\subsection{Towards a spacetime picture}

In the previous description, several steps relied heavily on the assumptions about the algebras being type $\mathrm{II}$. Specifically:

\begin{itemize}
    \item the existence of a trace $\mathrm{Tr}$ for the algebras;
    \item the existence of a tracial state $\ket{\Psi}$ associated with $\mathrm{Tr}$;
    \item the triviality of the modular operator and, consequently, the particularly simple form the modular conjugation takes. 
\end{itemize}
Since local algebras of QFT are type $\mathrm{III}$, it would be desirable to provide a generalization of the previous protocol to the type $\mathrm{III}$ scenario, which includes QFT.

\subsection{Kosaki-Longo index}

Jones construction has been extended to the case of type $\mathrm{III}$ factors in recent times \cite{Kosaki} \cite{Kosaki1} \cite{Longo} \cite{Longo1}. In particular, the notions of von Neumann inclusions, Jones extension and Jones index can be defined. In the following section, we shall review the basics of such theory following \cite{Kosaki} and apply it to the black hole information recovery protocol discussed above\footnote{The definition in \cite{Longo} is equivalent, but for the purposes of this paper Kosaki's approach is more fitting. Longo's formulation uses the crossed product which turns a type $\mathrm{III}$ algebra into a type $\mathrm{II}$ one; all the properties of Longo's notion of index are naturally inherited from the usual Jones index.}, attempting some viable generalizations for the type $\mathrm{III}$ case: this would better reflect the reality of an evaporating black hole in the semi-classical limit, as the local algebras of QFT are type $\mathrm{III}$. We will start reviewing the mathematics behind type $\mathrm{III}$ inclusions.

Given two type $\mathrm{III}$ factors $\mathcal{M}$ and $\mathcal{N}$ such that $\mathcal{N} \subset \mathcal{M}$, let us denote with $C(\mathcal{M}, \mathcal{N})$ the set of all faithful normal conditional expectations $\mathcal{E}: \mathcal{M} \rightarrow \mathcal{N}$. This set is included in the set of all faithful normal operator-valued weights from $\mathcal{M}$ to $\mathcal{N}$, denoted as $P(\mathcal{M}, \mathcal{N})$. Any faithful normal weight satisfies the bimodule property and is a multiple of a conditional expectation, when it is bounded. Since $\mathcal{N} \subset \mathcal{M} \implies \mathcal{M}' \subset \mathcal{N}'$, we can establish the canonical bijection

\begin{equation}
    (\cdot)^{-1} : \eta \in P(\mathcal{M}, \mathcal{N}) \rightarrow \eta^{-1} \in P(\mathcal{N}', \mathcal{M}')
\end{equation}
with the properties

\begin{itemize}
    \item $\big( \eta^{-1} \big)^{-1} = \eta$;
    \item $(\eta_1 \circ \eta_2)^{-1} = \eta_2^{-1} \circ \eta_1^{-1}$.
\end{itemize}
The bijection is briefly realized as follows \cite{Connes}. Let $\mathcal{M}$ be a von Neumann algebra on a Hilbert space $\mathcal{H}$ and $\mathcal{H}_{\Psi}$ the GNS Hilbert space constructed from the state $\Psi$ on $\mathcal{M}'$. If we consider the canonical injection $M \rightarrow \ket{M} \in \mathcal{H}_{\Psi}$ which realized the GNS construction, we may define a bounded operator $R_{\Psi}^{\ket{\xi}}: \mathcal{H}_{\Psi} \rightarrow \mathcal{H}$ such that

\begin{equation}
    R_{\Psi}^{\ket{\xi}} \ket{M} \eqdef M \ket{\xi} \text{.}
\end{equation}

\begin{Definition}
    The \textnormal{spatial derivative} $d\Phi/d\Psi$ is the unique positive self-adjoint operator $T$ on $\mathcal{H}$ satisfying

    \begin{equation}
        \bra{\xi} T \ket{\xi} = \overline{q}(\ket{\xi}) 
    \end{equation}
    where $\overline{q}(\ket{\xi})$ denotes the closure of $q(\ket{\xi}) \eqdef \Phi \Big( R_{\Psi}^{\ket{\xi}} {R_{\Psi}^{\ket{\xi}}}^* \Big)$.
    The \textnormal{canonical inverse} of $\Psi$ is the unique weight $\Psi^{-1}$ such that $d\Phi/d\Psi = \Phi \circ \Psi^{-1}$.
\end{Definition}

In many scenarios, given a state $\Psi$ on $\mathcal{M}$ and the state $\Psi' (X') \eqdef \Psi (J {X'}^* J)$ with $X' \in \mathcal{M}'$, we have $d\Phi/d\Psi' = \Delta_{\Phi|\Psi}$. 

\begin{Proposition}
    If $I \in \mathrm{Dom}(\mathcal{E}^{-1})$, then $\mathcal{E}^{-1}(I) \in \mathcal{M} \cap \mathcal{M}' = \mathbb{C} I$.
\end{Proposition}

\begin{proof}
    By the bimodule property, for any $M' \in \mathcal{M}'$,
    
    \begin{equation}
        \mathcal{E}^{-1}(I) = \mathcal{E}^{-1} \big( (M')^{-1} M' \big) = (M')^{-1} \mathcal{E}^{-1} (I) M'
    \end{equation}
    and thus $M' \mathcal{E}^{-1}(I) = \mathcal{E}^{-1}(I) M'$, which implies $\mathcal{E}^{-1}(I) \in \mathcal{M}'' = \mathcal{M}$. As a result, $\mathcal{E}^{-1}(I) \in \mathcal{M}' \cap \mathcal{M} =  \mathbb{C} I$.
\end{proof}

Consequently, we can identify $\mathcal{E}^{-1}(I)$ with a scalar.

\begin{Definition}
    Given a normal faithful conditional expectation $\mathcal{E}$, we define the \textnormal{index} of $\mathcal{E}$ as

    \begin{equation}
        \mathrm{Ind}(\mathcal{E}) \eqdef \mathcal{E}^{-1}(I) \in [1, \infty]
    \end{equation}
    where $\mathcal{E}^{-1}(I) \eqdef \infty$ whenever $I \not\in \mathrm{Dom}(\mathcal{E}^{-1})$.
\end{Definition}

If $\mathrm{Ind}(\mathcal{E}) < \infty$, we can normalize $\mathcal{E}^{-1}$ to obtain a conditional expectation

\begin{equation}
    \mathcal{E}' \eqdef \big( \mathrm{Ind}(\mathcal{E}) \big)^{-1} \mathcal{E}^{-1} \in C(\mathcal{N}', \mathcal{M}')
\end{equation}
which satisfies $(\mathcal{E}')^{-1} = \mathrm{Ind}(\mathcal{E}) \ \big( \mathcal{E}^{-1} \big)^{-1} = \mathrm{Ind}(\mathcal{E}) \ \mathcal{E}$, thus $\mathrm{Ind}(\mathcal{E}') = \mathrm{Ind}(\mathcal{E})$.

\begin{Definition}
    Whenever there is a conditional expectation with finite index associated to an inclusion of factors $\mathcal{N} \subset \mathcal{M}$, the unique conditional expectation $\mathcal{E}_0 \in C(\mathcal{M}, \mathcal{N})$ such that

    \begin{equation}
        \mathrm{Ind}(\mathcal{E}_0) = \inf_{\mathcal{E} \in C(\mathcal{M}, \mathcal{N})} \mathrm{Ind}(\mathcal{E}) \eqdef [\mathcal{M} : \mathcal{N}]
    \end{equation}
    is called \textnormal{minimal conditional expectation}. $[\mathcal{M} : \mathcal{N}]$ is referred to as the \textnormal{index} of $\mathcal{N}$ in $\mathcal{M}$.
\end{Definition}

It can be proved that such $\mathcal{E}_0$ is unique \cite{Connes}. Furthermore, if $\mathcal{E}_0$ is minimal in $C(\mathcal{M}, \mathcal{N})$, then $\mathcal{E}_0'$ is minimal in $C(\mathcal{N}', \mathcal{M}')$.

Consider now an inclusion $\mathcal{N} \subset \mathcal{M}$ with finite index. The following theorem proves that the relative commutant $\mathcal{N}' \cap \mathcal{M}$ is finite-dimensional: as a result, even though in our case $\mathcal{M}$ and $\mathcal{N}$ are type $\mathrm{III}$, Alice's (as well as Bob's) algebra admits a trace.

\begin{Proposition}
    Consider the inclusion $\mathcal{N} \subset \mathcal{M}$ where $\mathcal{M}$ and $\mathcal{N}$ are both type $\mathrm{III}$ factors. If $[\mathcal{M} : \mathcal{N}] < \infty$, $\mathcal{A} = \mathcal{N'} \cap \mathcal{M}$ is finite-dimensional.
\end{Proposition}

\begin{proof}
    The proof relies on the property
    \begin{equation}
        [\mathcal{M}:\mathcal{N}]^{\frac{1}{2}} = \sum_i [\mathcal{M}_i:\mathcal{N}_i]^{\frac{1}{2}}
        \label{index sum}
    \end{equation}
    with $\mathcal{N} P_i \subset \mathcal{M}_i \eqdef P_i \mathcal{M} P_i$, valid for any partition $\sum_i P_i = I$ of the identity by projections $P_i$ in $\mathcal{N}' \cap \mathcal{M}$.
    If $\mathcal{N}' \cap \mathcal{M}$ were infinite-dimensional, we could find arbitrarily large partitions, yielding $[\mathcal{M}:\mathcal{N}] = \infty$ due to \ref{index sum}.
\end{proof}

In light of that, inclusions of factors of local algebras arising from spacetime region inclusions due to the isotonic property cannot have finite index, otherwise we would have a finite-dimensional algebra instead of another type $\mathrm{III}$ one. If we still ignore that, for example assuming that Alice's diary is so simple that it essentially encodes a few qubits' information, then we may define a trace for both Alice's and Bob's algebras. We will discuss a concrete physical realization of such scenario below.

In order to implement the algorithm as before, we need to find a generalization of \ref{Jind}.

\begin{Proposition}
    Given $\mathcal{E} \in C(\mathcal{M}, \mathcal{N})$, $\left. \mathcal{E} \right|_{\mathcal{N}' \cap \mathcal{M}}$ has range in $\mathbb{C} I$.
    \label{expectation scalar}
\end{Proposition}

\begin{proof}
    Let us show that $\mathcal{E}(N') \in \mathcal{N}'$, for any $N' \in \mathcal{N}'$.

    \begin{equation}
        N \mathcal{E}(N') = \mathcal{E}(N N') = \mathcal{E} (N' N) = \mathcal{E}(N') N \text{.}
    \end{equation}
    As a result, $\mathcal{E}(N') \in \mathcal{N} \cap \mathcal{N}' = \mathbb{C}I$ since $\mathcal{N}$ is a factor.
\end{proof}

Moreover, when we use the minimal conditional expectation $\mathcal{E}_0$, its uniqueness implies $\mathcal{E}_0 \circ \mathrm{Ad}_{U} = \mathcal{E}_0$ for any unitary $U \in \mathcal{N}' \cap \mathcal{M}$\footnote{$U$ must belong to $\mathcal{A}$ otherwise we would have $\mathcal{E}_0(M) = \mathcal{E}_0(U^*MU) = U^* \mathcal{E}_0(M) U \neq \mathcal{E}_0(M) U^* U = \mathcal{E}_0(M)$.}, hence $\mathcal{E}_0 (AM) = \mathcal{E}_0(MA), \ \forall A \in \mathcal{A}, \ \forall M \in \mathcal{M}$, yielding a trace on $\mathcal{A} = \mathcal{N}' \cap \mathcal{M}$ \cite{Longo}. The following property, first proved by Jones \cite{Jones} for type $\mathrm{II}_1$ factors, provides a crucial step towards the generalization of the previous results.

In the type $\mathrm{III}$ case, although no finite projections exist, we can still define Jones projections through the conditional expectation. Specifically, we can define the Jones projection $e_{\mathcal{N}}$ associated to the inclusion $\mathcal{N} \subset \mathcal{M}$ as the projection such that $e_{\mathcal{N}} M e_{\mathcal{N}} = \mathcal{E}_0(M) e_{\mathcal{N}}, \ \forall M \in \mathcal{M}$. We can represent all algebras in the largest algebra GNS representation $\mathcal{H}$ and $e_{\mathcal{N}}$ projects from $\overline{\mathcal{M}\ket{\Omega}}$ to $\overline{\mathcal{N}\ket{\Omega}}$, where $\ket{\Omega}$ is a common cyclic and separating vector. The fact that the projection is infinite in the Murray-von Neumann sense does not interfere with such definition.

\begin{Proposition}
    The relation

    \begin{equation}
        \mathcal{E}_0(e_{{\mathcal{N}}_1}) = [\mathcal{M} : \mathcal{N}]^{-1} I
    \end{equation}
    holds, with $e_{{\mathcal{N}}_1}$ being the Jones projection associated with the inclusion $\mathcal{N}_1 \subset \mathcal{N}$\footnote{$\mathcal{N}_1$ belongs to the \textit{downward} Jones basic construction. The Jones tower consists of the inclusions $...\subset \mathcal{N}_j \subset ... \subset \mathcal{N}_1 \subset \mathcal{N} \subset \mathcal{M} \subset \mathcal{M}_1 \subset ... \subset \mathcal{M}_k \subset ...$.}.
    \label{Index inv}
\end{Proposition}

The proof for the general case can be found in \cite{Longo}.

With these mathematical tools at our disposal, we can move onto generalizing equation \ref{fResult} to the field-theoretic case.

\section{The protocol for charged black holes}
\subsection{Physical setup}

We consider a toy model for the adiabatic evaporation of a hadronically-charged black hole, described in terms of the DHR superselection theory.
To realize black hole evaporation, one typically needs to work at finite $G_N$ in order to allow for mass loss inside the horizon. This description could switch the involved algebras from type $\mathrm{III}$ to type $\mathrm{II}$ due to quantum gravity corrections \cite{Witten-Longo}. Nevertheless, one could in principle take into account gravitational corrections and describe the picture semiclassically: the net effect of $G_N$ is a deformation to the black hole geometry into an evaporating one; if we focus just on the radiation degrees of freedom, the algebra describing them is still type $\mathrm{III}$, as any local algebra of a QFT in curved backgrounds.

Naturally, since the reference to black hole geometry is never explicitly present in the following discussion, we refer to \textit{black holes} because of the causal structure of the observables, which resembles the one we would expect from a realistic black hole scenario (see figure \ref{Black hole algebras}). Specifically, by \textit{charged black hole} we hereby mean a QFT endowed with a DHR endomorphism localized in a certain region $\mathcal{O}$. The region $\mathcal{O}$ should represent the black hole interior, since DHR endomorphisms create a charged state localized in the corresponding region; hence we interpret it as the state describing the charged matter forming the black hole. 

By \textit{adiabatic} we mean that the horizon shrinks slowly, so that we basically keep the same localization for the field algebra $\mathcal{A}(\mathcal{O})$ describing the black hole interior both before and after an evaporation step. This way, the black hole undergoes a change in the charge quantum number, but keeps the same localization region in spacetime. The term \textit{evaporating} may be misleading, since strictly speaking the horizon is fixed; what is being evaporated out are just charge quanta.
The obstruction to working with local algebras coming from spacetime inclusions is that they do not have finite index and consequently cannot be dealt with using the Jones index machinery. Still, this does not provide a no-go theorem and a similar protocol could, in principle. be found even for spacetime regions.

An observation to be made is that this picture is \textit{dynamical} in the sense that the algebra of the black hole does indeed lose some bits of information, due to charge change. However, we neglect a possible $U$ modelling the time evolution of the system: from when Alice creates her state to when Bob tries to reproduce it, the state undergoes a time evolution. It was argued in \cite{Verlinde} that including such $U$ could in principle be done and Bob would have access to an operator which approximates $U^{-1}$ (undoing the evolution) up to a certain error, but explicit calculations have not been done yet for the type $\mathrm{II}$ or $\mathrm{III}$ case. Summing up, the dynamics of the evaporation is encoded just in the fact that we deal with an inclusion $\mathcal{N} \subset \mathcal{M}$, specifically due to a change in the black hole charge.

A more mathematically precise picture is provided by a charged state localized in a causal diamond in Minkowski spacetime undergoing a superselection sector change (see figure \ref{Minkowski}). By the equivalence principle, we expect an analogous behaviour for charged black holes.

\begin{figure}[h!]
     \centering
     \begin{subfigure}[b]{0.475\textwidth}
         \centering
         \includegraphics[width=\textwidth]{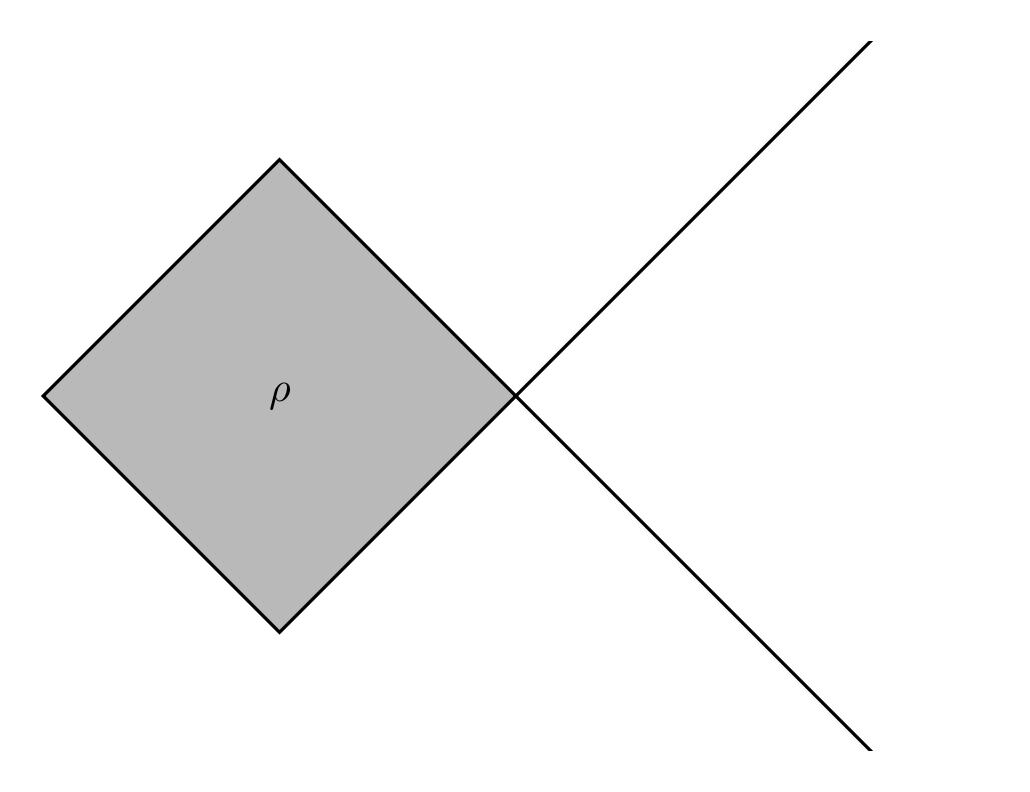}
         \caption{A charged state in the causal diamond is described by a DHR endomorphism $\rho$.}
         \label{fig:first}
     \end{subfigure}
     \hfill 
     \begin{subfigure}[b]{0.475\textwidth}
         \centering
         \includegraphics[width=\textwidth]{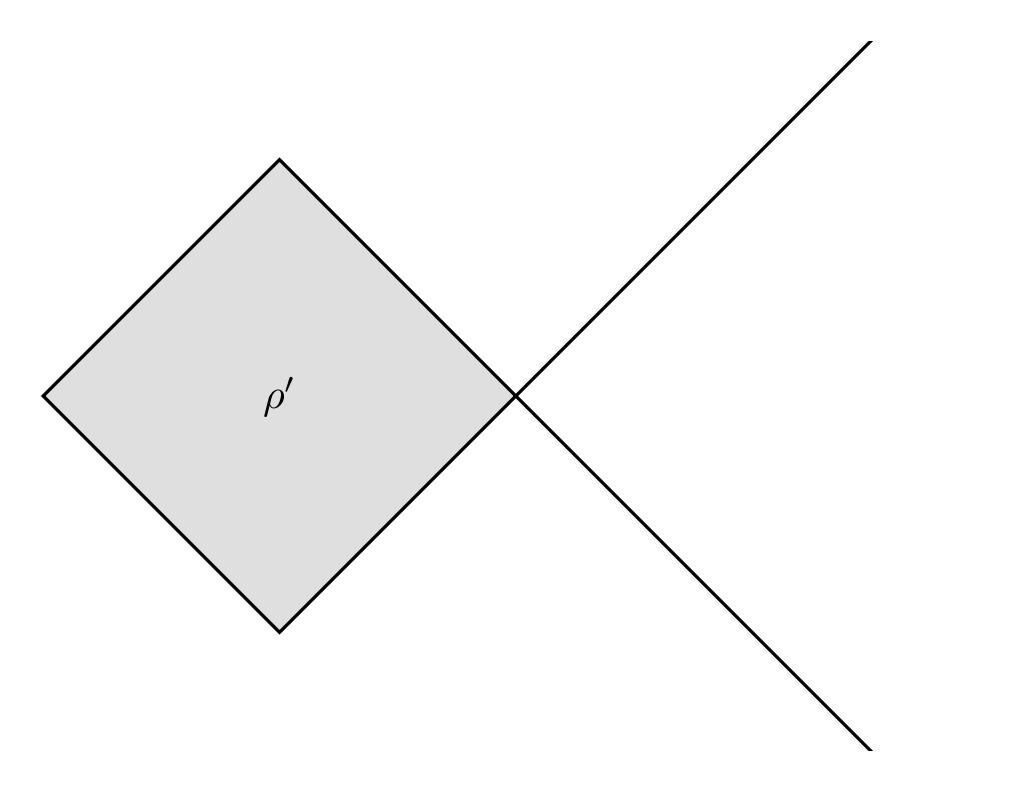}
         \caption{Another charged state in the causal diamond is described by another DHR endomorphism $\rho'$.}
         \label{fig:second}
     \end{subfigure}
     
     \caption{Visualization of the physical scenario in Minkowski spacetime. If the state somehow loses some charge (for black holes, due to evaporation), the new state is described by another endomorphism $\rho'$. We may think of $\rho$ as $\sigma \circ \rho'$, where $\sigma$ represents the missing charge.}
     \label{Minkowski}
\end{figure}

Concerning the black hole information paradox, we want to emphasise that this protocol does not directly address the information loss issue, since we are working with a fixed horizon. Even if this was not the case, the information paradox arises as the black hole reaches complete evaporation, leaving just thermal radiation outside, or at the very least as it approaches the Page time, where unitarity contrasts with the semiclassical result since radiation entropy starts decreasing. Far before the Page time, one may expect standard QFT in curved backgrounds to hold as an effective field theory.

Notwithstanding the strong limitations of the protocol, we remark that the extension of \cite{Verlinde} to the type $\mathrm{III}$ case, even keeping the index finite, is non-trivial at all, since the previous formulas allowing the teleportation explicitly rely on the existence of a trace on $\mathcal{M}$. Moreover, the generality of the mathematics adopted also allows for applications to different physical scenarios, which may have nothing to do with gravity. A simple example could be a single qubit coupled to an environment $\mathcal{N}$ consisting of a spin chain in the thermodynamic limit built as in \cite{Witten}. The algebra $\mathcal{M}$ would be $\mathcal{N} \otimes M_2(\mathbb{C})$ and both $\mathcal{N}$ and $\mathcal{M}$ would be type $\mathrm{III}$. The relative commutant $\mathcal{A}$ would correspond precisely to the qubit's algebra $M_2(\mathbb{C})$, which is finite-dimensional, hence allowing for a finite value of the Jones index. Alternatively, we may also think of $\mathcal{M} = \mathcal{N} \otimes M_n(\mathbb{C})$ as the algebra describing the fields in the black hole interior plus some qubits that had been tossed inside and $\mathcal{N} \cong \mathcal{N} \otimes I$ as the final field algebra where all the residual information about the qubits has already been evaporated out. Even in this scenario, the index is finite and the protocol applies smoothly.

\subsection{Type $\mathrm{III}$ version of the protocol}

To generalize the protocol formula \cite{Verlinde}, let us consider the state $\mathcal{E}_0$ on the relative commutant $\mathcal{A} = \mathcal{N}' \cap \mathcal{M}$: because of the tracial property, we have that such state on Alice's observables yields

\begin{equation}
    \mathcal{E}_0 (A) = \mathrm{Tr}_{\mathcal{A}}(A) I = \mathrm{Tr}_{\mathcal{A}}(A)
    \label{A}
\end{equation}
where $\mathrm{Tr}_{\mathcal{A}}$ denotes the canonical trace on Alice's algebra. Notice that, in general, the trace on $\mathcal{A}$ is not unique. To convince ourselves of this, consider the finite-dimensional algebra

\begin{equation}
    \mathcal{A} = M_n(\mathbb{C}) \oplus M_n(\mathbb{C}) \text{.}
\end{equation}
We may define a trace as

\begin{equation}
    \mathrm{Tr}_{\lambda}(A \oplus B) \eqdef \lambda \ \mathrm{Tr}(A) + (1-\lambda) \ \mathrm{Tr}(B)
    \label{lambda}
\end{equation}
and any choice of $\lambda \in [0, 1]$ would be admissible. The canonical trace is the one induced by $\mathcal{E}_0$, as described above. 

For example, if we consider the inclusion $\mathcal{N} \subset \mathcal{M}$ with $\mathcal{M} = \big( \mathcal{R} \otimes M_n(\mathbb{C}) \big) \oplus \big( \mathcal{R} \otimes M_n(\mathbb{C}) \big)$ and $\mathcal{N} = \big( \mathcal{R} \otimes I \big) \oplus \big( \mathcal{R} \otimes I \big)$, with $\mathcal{R}$ being a generic type $\mathrm{III}$ factor, we have

\begin{equation}
    \big( \mathcal{R} \otimes I \big)' \cap \big( \mathcal{R} \otimes M_n(\mathbb{C}) \big) = \big( \mathcal{R}' \otimes I \big) \cap \big( \mathcal{R} \otimes M_n(\mathbb{C}) \big) = I \otimes M_n(\mathbb{C}) \text{;}
\end{equation}
as a result, $\mathcal{N}' \cap \mathcal{M} = M_n(\mathbb{C}) \oplus M_n(\mathbb{C})$. The class of all conditional expectations is obtained considering $\mathcal{E}_{\varphi}(X \otimes A) = X \otimes \varphi(A) I$ for any weight $\varphi$\footnote{Since $\mathcal{N}$ contains two copies of elements in the form $X \otimes I$, the only way to preserve the second member of the tensor product is to transform $A$ into a scalar. We did not put $\mathcal{E}_{\varphi}(X \otimes A) = \varepsilon(X) \otimes \varphi(A) I$ for some conditional expectation $\varepsilon$ on $\mathcal{R}$ because we chose the first term of the tensor product to be $\mathcal{R}$ both for $\mathcal{M}$ and $\mathcal{N}$.}. The corresponding minimal conditional expectation $\mathcal{E}_0$ is

\begin{equation}
    \mathcal{E}_0 \big( (X \otimes A) \oplus (Y \otimes B) \big) = \frac{1}{2} \Big( \big( X \otimes \mathrm{Tr}(A) I) \oplus \big( Y \otimes \mathrm{Tr}(B) I \big) \Big) \text{.}
\end{equation}
In fact, when we consider $\varphi(A) = \mathrm{Tr}(\rho_{\varphi}A)$, the Pimsner-Popa bound holds \cite{Pimsner-Popa}:

\begin{equation}
    \mathcal{E} (X^+ \otimes A^+) \geq \big( \mathrm{Ind}(\mathcal{E}) \big)^{-1} (X^+ \otimes A^+), \ \forall X^+ \in \mathcal{R}^+, \ \forall A^+ \in M_n(\mathbb{C})^+
\end{equation}
with the abuse of notation of denoting again with $\mathcal{E}$ its restriction to the first term of the direct sum. Hence,

\begin{equation}
    \varphi(A^* A) \geq \big( \mathrm{Ind}(\mathcal{E}) \big)^{-1} A^* A, \ \forall A \in M_n(\mathbb{C}) \text{.}
\end{equation}
Since we may decompose $\rho = \sum_{k = 1}^n p_k \ket{k} \bra{k}$, it easily follows\footnote{It can be seen choosing $A = \ket{j} \bra{j}$, so that $p_j I \geq \big( \mathrm{Ind}(\mathcal{E}) \big)^{-1} \ket{j} \bra{j}$, which implies $p_j \geq \big( \mathrm{Ind}(\mathcal{E}) \big)^{-1}$, meaning that $\mathrm{Ind}(\mathcal{E}) \leq \frac{1}{\min_{j}{p_j}}$.} that the index is minimal for the choice $\rho_{\varphi} = \frac{1}{n} I$. Consequently, the canonical trace is the one corresponding to the choice $\lambda = \frac{1}{2}$ in equation \ref{lambda}. Notice that such choice of the trace corresponds to the one used in the computation of von Neumann entropy.

From Alice's point of view, \ref{A} represents a maximally mixed state, i.e. the state with maximum von Neumann entropy. It would be interesting to see what happens when changing the inclusion $\mathcal{N} \subset \mathcal{M}$ giving rise to the same $\mathcal{A}$, namely whether the canonical trace corresponds still to the one required in von Neumann entropy, possibly adopting the framework discussed in \cite{Balachandran}.

Concerning Bob's algebra, we have a canonical trace induced by the minimal conditional expectation relative to the inclusion $\mathcal{M}_1 \subset \mathcal{M}_2$, denoted here as $E_0$:

\begin{equation}
    E_0(B) = \mathrm{Tr}_{\mathcal{B}}(B), \ \forall B \in \mathcal{B} \text{.}
\end{equation}
From the proposition \ref{Index inv}, we know that

\begin{equation}
    [\mathcal{M}_2 : \mathcal{M}_1] = \frac{1}{E_0(e)}
\end{equation}
with $e \eqdef e_{\mathcal{M}}$ being the Jones projection associated with the inclusion $\mathcal{M} \subset \mathcal{M}_1$. Notice that $e \in \mathcal{M}_2$, consistently with the domain of $E_0$. Furthermore, we have $[\mathcal{M}:\mathcal{N}] = [\mathcal{M}_1:\mathcal{M}] = [\mathcal{M}_2:\mathcal{M}_1]$ \cite{Longo1}, yielding

\begin{equation}
    \begin{split}
        E_0 \big( e \Gamma(A) \big) &= E_0 \big( e^2 \Gamma(A) \big) = E_0 \big( e \Gamma(A) e \big) \\
        &= E_0 \big( E_0 ( \Gamma(A) ) e \big) = \mathrm{Tr}_{\mathcal{B}}\big( \Gamma(A) \big) E_0 (e) \\
        &= [\mathcal{M}:\mathcal{N}]^{-1} \mathrm{Tr}_{\mathcal{B}}\big( \Gamma(A) \big) \text{.}
    \end{split}
\end{equation}
The core property for the next step is achieved thanks to the following result.

\begin{Proposition}
    The canonical shift $\Gamma$ preserves the canonical trace.
    \label{Preserved trace main}
\end{Proposition}

\begin{proof}
    The restriction of a conditional expectation on the relative commutant yields a state on it, as evident from \ref{expectation scalar}: this one-to-one correspondence relates normal expectations with normal states on the relative commutant \cite{Longo}. The minimal conditional expectation corresponds to the canonical trace and we can consequently associate an index to each trace. The canonical conjugation, being an anti-isomorphism, preserves the index; as a result, it must map the canonical trace into a canonical trace.
\end{proof}

A theorem directly implying \ref{Preserved trace main} is present in \cite{Longo3} and it hold for both implications. Finally, we have

\begin{equation}
    \begin{split}
        E_0 \big( e \Gamma(A) \big) &= [\mathcal{M}:\mathcal{N}]^{-1} \mathrm{Tr}_{\mathcal{B}}\big( \Gamma(A) \big) = [\mathcal{M}:\mathcal{N}]^{-1} \mathrm{Tr}_{\mathcal{A}}\big( A \big) \\
        &= [\mathcal{M}:\mathcal{N}]^{-1} \mathcal{E}_0(A)
    \end{split}
    \label{teleport}
\end{equation}
providing a generalization for the type $\mathrm{III}$ case of the formula \ref{fResult}. From Alice's point of view it holds for any state at her disposal: following the same discussion in \cite{Verlinde}, we can choose instead of $A$ any operator in the form $X^* A X$, with $X$ being an isometry. $X$ has the role of creating a state $\ket{X}$ in Alice's GNS Hilbert space and $A$ is a generic observable; when the teleportation takes place, Bob applies $\Gamma$ to $X^* A X$. From Bob's point of view, the state induced by $E_0$ also corresponds to the maximally mixed state in his GNS Hilbert space. So Bob needs to perform an operation that turns the state of the system is in into his maximally entangled state. Hence the formula above, although generalizing Markov's property to the type $\mathrm{III}$ case, does not directly implement the protocol. In fact, conditional expectations are states only when restricted to the relative commutants. Let us consider a state vector $\ket{\Psi}$ in the largest GNS Hilbert space such that, when restricted on $\mathcal{A}$, yields the canonical trace: $\bra{\Psi} A \ket{\Psi} = \mathrm{Tr}_{\mathcal{A}}(A) = \mathcal{E}_0(A)$. To have $E_0$ appear, we apply the Jones projection $e_{\mathcal{M}_1}$ to such vector, so that

\begin{equation}
    \begin{split}
        \bra{\Psi} e_{\mathcal{M}_1} \Gamma (A) e_{\mathcal{M}_1} \ket{\Psi} &= \bra{\Psi} E_0 \big( \Gamma (A) \big) e_{\mathcal{M}_1} \ket{\Psi} = \mathcal{E}_0 (A) \bra{\Psi} e_{\mathcal{M}_1} \ket{\Psi} \\
        &= ||\ket{e_{\mathcal{M}_1} \Psi}||^2 \bra{\Psi} A \ket{\Psi} \text{.}
    \end{split}
    \label{teleport final}
\end{equation}
If we prove that $||\ket{e_{\mathcal{M}_1} \Psi}||^2 = [\mathcal{M}:\mathcal{N}]^{-1}$, we have successfully implemented the desired type $\mathrm{III}$ generalization of the protocol described in \cite{Verlinde}. That is easily realized if we choose $\bra{\Psi} \cdot \ket{\Psi} = \Phi \circ \tilde{E}_0(\cdot)$, according to \ref{Index inv}, where $\tilde{E}_0: \mathcal{M}_3 \rightarrow \mathcal{M}_2$ is the corresponding minimal conditional expectation and $\Phi$ a suitable state. When restricted to Alice's algebra, $\tilde{E}_0(\cdot)$ can yield $\mathrm{Tr}_{\mathcal{A}}(\cdot)$, since we still have freedom over $\Phi$\footnote{Since $\left. \tilde{E}_0 \right|_{\mathcal{A}} = \mathrm{Id}$, we can pick a state $\Phi$ that yields $\mathrm{Tr}_{\mathcal{A}}$ when restricted on $\mathcal{A}$.}. Notice that $e_{\mathcal{M}_1} \not\in \mathcal{B}$, so we need a third party in order to achieve this operation, which however consists in just one step which is independent of the observable $A$ in question.

So far, the discussion applies generally, without specific reference to charged black holes.

\subsection{DHR superselection sectors}

\subsubsection{Superselection rules}

In quantum physics, we often have situations where only incoherent superposition of some states is possible \cite{Superselection}. As a pedagogical example \cite{Giulini}, consider two different total-spin states $\ket{s=1}$ and $\ket{s=1/2}$. We know that superpositions between such states are forbidden and the reason is that they transform differently under a $2 \pi$-rotation:

\begin{equation}
    R(2 \pi) \frac{1}{\sqrt{2}} \big( \ket{s=1} + \ket{s=1/2} \big) = \frac{1}{\sqrt{2}} \big( \ket{s=1} - \ket{s=1/2} \big) \text{.}
\end{equation}
If we still try and superposing those states into $\ket{\psi} = \alpha \ket{s=1} + \beta \ket{s=1/2}$, we need to impose $\bra{\psi} A \ket{\psi} = \bra{\psi} R(2 \pi)^* A R(2 \pi) \ket{\psi}$, since a $2 \pi$-rotation should not change the results of experiments:

\begin{equation}
    \begin{split}
        |\alpha|^2 \langle A \rangle_{s=1} + |\beta|^2 \langle A \rangle_{s=1/2} + 2 \mathfrak{Re} \big( \overline{\alpha} \beta \bra{s=1} A \ket{s=1/2} \big) \\
        = |\alpha|^2 \langle A \rangle_{s=1} + |\beta|^2 \langle A \rangle_{s=1/2} - 2 \mathfrak{Re} \big( \overline{\alpha} \beta \bra{s=1} A \ket{s=1/2} \big)
    \end{split}
    \label{superselection}
\end{equation}
for any observable\footnote{Notice that this holds for \textit{observables}, i.e. gauge-invariant operators.} $A$ of our system and arbitrary $\alpha$, $\beta$. As a consequence, a superselection rule takes place, telling us that the constraint $\bra{s=1} A \ket{s=\frac{1}{2}} = 0$ must hold for any observable $A$. This means that if we try to coherently superpose $\ket{s=1}$ and $\ket{s=1/2}$, we actually end up with the equivalent description of a mixed state $\rho = |\alpha|^2 \ket{s=1} \bra{s=1} + |\beta|^2 \ket{s=1/2} \bra{s=1/2}$ for the two possible outcomes, as evident from

\begin{equation}
    \bra{\psi} A \ket{\psi} = |\alpha|^2 \langle A \rangle_{s=1} + |\beta|^2 \langle A \rangle_{s=1/2} \text{,}
\end{equation}
which in algebraic terms tells us that the state $\psi$ is a convex combination of $s=1$ and $s=1/2$. As a result, the superposition is actually classical\footnote{Trying to superpose two different total-spin states is equivalent to stating our ignorance about the actual total-spin state of the system.} or \textit{incoherent} \cite{Decoherence}. This also tells us that the representations of the algebra of observables $\mathcal{A}$ arising from those states lie in inequivalent Hilbert spaces and can be split into a direct sum.

This phenomenon happens in general when we deal with representations of symmetries. The line of reasoning is analogous to the one sketched above. Another example discussed by Weinberg \cite{Weinberg} is the case of the Galilean group in non-relativistic quantum mechanics, where we cannot superpose states of different mass. Since representations are generally projective \cite{Wigner}, we are allowed to have a phase relating $U(g_1 g_2)$ and $U(g_1)U(g_2)$. If we consider boosts generated by $\vec{K}$ and translations generated by $\vec{P}$, it holds that

\begin{equation}
    e^{i \vec{v} \cdot \vec{K}} e^{-i \vec{a} \cdot \vec{P}} = e^{i\frac{M}{2} \vec{a} \cdot \vec{v}} e^{i (\vec{v} \cdot \vec{K} - \vec{a} \cdot \vec{P})}
\end{equation}
forbidding coherent superpositions of states with different mass $M$.

\subsubsection{Motivation from short-range interactions and charges}

When we deal with theories carrying charges \cite{Aharonov} \cite{Strocchi}, we also deal with superselection sectors. For instance, already at the non-relativistic level gauge transformations $\Lambda \in U(1)$ act on charged particle states $\ket{\psi}$ as

\begin{equation}
    \ket{\psi} \rightarrow  e^{i q \Lambda(\vec{x})} \ket{\psi} \text{.}
\end{equation}
In QFT, the arising inequivalent representations can be connected by suitable intertwiners which form the \textit{field algebra} $\mathcal{F}(\mathcal{O})$. The algebra of observables $\mathcal{A}(\mathcal{O}) = \mathcal{F}(\mathcal{O})^{\alpha_G}$ is interpreted as the fixed-points subalgebra of $\mathcal{F}(\mathcal{O})$ with respect to the action $\alpha_G$ of the gauge group $G$.

The DHR treatment of superselection sectors focuses on charges associated with short-ranged interactions \cite{DHR}. The electric charge, due to Gauss law, can always be detected far away from the source, while short-ranged interactions, such as the one yielding the Yukawa potential $V \propto - \frac{e^{-Mr}}{r}$ \cite{Yukawa}, decreases rapidly far from the source, allowing for a description equivalent to the vacuum one. An explicit example of such theory is given by the pseudoscalar-coupling Lagrangian density \cite{Gell-Mann}:

\begin{equation}
    \mathcal{L} = \frac{1}{2} \partial_{\mu} \pi^a \partial^{\mu} \pi_a - \frac{1}{2} M^2 \pi^a \pi_a + \overline{N} (i \slashed{\partial} - m) N - i g \overline{N} \gamma^5 \tau^a \pi_a N
\end{equation}
where $N = (p, n)^T$ denotes nucleon field doublet, $\pi^a$ are the pion fields, $\tau^a$ are the Pauli matrices generating $SU(2)$ and we assume exact isospin symmetry, $m_{\text{proton}} = m_{\text{neutron}} = m$. This interaction gives rise to a Yukawa potential \cite{Nucleons} \cite{Peskin}. Generally, interactions mediated by massive particles give rise to short-ranged potentials.

\subsubsection{DHR theory and statistical dimension}

The formal treatment of superselection sectors in Algebraic Quantum Field Theory has been initiated by Doplicher, Haag and Roberts \cite{DHR}. The general assumption is that the charges are short-ranged: if $\pi_0$ is the vacuum GNS representation of the algebra of observables, we consider representations $\pi$ constructed from charged states (\textit{charged sectors}) so that the condition

\begin{equation}
    \left. \pi_0 \right|_{\mathcal{A}(\mathcal{O}')} \cong \left. \pi \right|_{\mathcal{A}(\mathcal{O}')}
\end{equation}
is fulfilled for a sufficiently large double cone $\mathcal{O}$\footnote{Double cone here means a causal diamond.}; namely, there is a unitary $V$ such that

\begin{equation}
    \pi_0 (A) = V^* \pi (A) V \text{.}
\end{equation}
For our purposes, we could also identify $\mathcal{A}(\mathcal{O})$ with $\pi_0 \big(\mathcal{A}(\mathcal{O}) \big)$. This naturally suggests the definition of a map

\begin{equation}
    \rho (A) \eqdef V^* \pi(A) V
\end{equation}
which is the identity in $\mathcal{U}'$ and is non-trivial in $\mathcal{U}$. We say that $\rho$ is \textit{localized} in $\mathcal{U}$.

If Haag duality holds, then for any $A \in \mathcal{A}(\mathcal{O})$ and $B \in \mathcal{A}(\mathcal{O}')$,

\begin{equation}
    \rho (A) B = \rho(AB) = \rho(BA) = B \rho(A) \implies \rho \big( \mathcal{A}(\mathcal{O}) \big) \subseteq \mathcal{A}(\mathcal{O}')' = \mathcal{A}(\mathcal{O})
\end{equation}
thus $\rho$ is an endomorphism. Concretely, we have

\begin{equation}
    \pi \cong \pi_0 \circ \rho
\end{equation}
and in particular $\Psi_{\rho} = \Omega \circ \rho$ for any charged state $\Psi_{\rho}$: $\rho$ can be interpreted as "creating" the charge.

Before moving on, let us present the physical meaning of the endomorphisms $\rho$ through an example \cite{Casini}. If we consider a Dirac field, physical observables are constructed taking bilinear in the fields so that they commute at spacelike-separated regions. The algebra of observables is thus the algebra $\mathcal{A}$ generated by an even number of Dirac fields $\mathbb{I}$, $\psi(x) \psi(y)$, $\psi^*(x) \psi(y)$, $\psi^*(x) \psi^*(y)$ and so on, spread with test functions having support inside $\mathcal{O}$. Given a charged\footnote{The charge here is the fermion number.} state $\ket{\Psi} = V_{\mathcal{O}} \ket{\Omega}$, with $V_{\mathcal{O}}^* \eqdef \int_{\mathcal{O}} d^d x f(x) \left( \psi^*(x) + \psi(x) \right)$ and $f(x)$ a real spinor supported in $\mathcal{O}$, any expectation value $\bra{\Psi} A \ket{\Psi}$ in the 1-fermion superselection sector can be equivalently computed as $\bra{\Omega} V_{\mathcal{O}}^* A V_{\mathcal{O}} \ket{\Omega} = \bra{\Omega} \rho_{\mathcal{O}}(A) \ket{\Omega}$ in the vacuum representation, defining $\rho_{\mathcal{O}}(\cdot) \eqdef V_{\mathcal{O}}^* (\cdot) V_{\mathcal{O}}$.
A similar reasoning can be adapted to curved spacetimes, where now $\mathcal{O}$ represents a black hole region. If the state inside $\mathcal{O}$ is charged, then there is a suitable endomorphism $\rho$ localized in the black hole which describes the charged state in terms of the vacuum representation (assuming a suitable \textit{vacuum} can be defined in such geometry). If during the evaporation the charge inside the black hole changes due to particle production, its interior may carry a different representation $\rho'$. This models the situation described above for the inclusion $\mathcal{N} \subset \mathcal{M}$, where $\mathcal{N}$ and $\mathcal{M}$ describe two differently-charged sectors of field algebra inside the horizon.

There may be some situations where we have the same charged state but differently distributed in space. Given two localized endomorphisms $\rho_1$ and $\rho_2$, they are said to be \textit{equivalent} ($\rho_1 \sim \rho_2$) if $\pi_0 \circ \rho_1 \cong \pi_0 \circ \rho_2$, that is there exists a unitary $U$ such that $\rho_1 = \mathrm{Ad}_U \circ \rho_2$. For instance, two particle states with the same charge localized in spacelike-separated regions correspond to equivalent localized endomorphisms: $U$ would be a translation in such scenario. The operator

\begin{equation}
    \epsilon_{\rho} \eqdef U^* \rho(U)
\end{equation}
has the role of swapping $\rho \circ \rho'$ and $\rho' \circ \rho$, i.e. it acts as a unitary representation of the permutation group. Indeed,

\begin{equation}
    \epsilon_{\rho} \big( \rho \circ \rho' (A) \big) \epsilon_{\rho}^* = U^* \rho(U) \rho \big( \rho'(A) \big) \rho(U^*) U = U^* \rho \big( U \rho'(A) U^* \big) U = \rho' \circ \rho (A) \text{.}
\end{equation}
It can be proven that, for $\rho$ irreducible \cite{DHR}, we can introduce a left inverse of the endomorphism $\rho$, namely a map $\varphi$ such that $\varphi \circ \rho = \mathrm{Id}$, satisfying the property

\begin{equation}
    \varphi(\epsilon_{\rho}) \in \{ 0 \} \cup \Big\{ \pm \frac{1}{d} \ \bigg| \ d \in \mathbb{N} \Big\} \text{.}
\end{equation}
The quantity $d = |\varphi(\epsilon_{\rho})|^{-1} \in \mathbb{N} \cup \{ \infty \}$ is called the \textit{statistical dimension} of the superselection sector $\rho$; the name comes form the fact that it coincides with the dimension of the representation of the permutation group associated with $\rho$.

The following result explicates a link between statistical dimension and Jones index, which (as shown earlier) provides an information-theoretic nature for the teleportation protocol.

\begin{Proposition}[\textbf{Index-statistics theorem}]
    The following equality holds:
    
    \begin{equation}
        [\mathcal{A}(\mathcal{O}):\rho \big( \mathcal{A}(\mathcal{O}) \big)] = \big( d(\rho) \big)^2
        \label{index-statistics}
    \end{equation}
\end{Proposition}

The proof can be found in \cite{Longo}; the reader may be heuristically convinced about the validity of \ref{index-statistics} comparing the methods used by Kosaki in its definition of the type-$\mathrm{III}$ index and the DHR theory of the statistical dimension.

The statistical dimension can be interpreted at a field-theoretic level as follows. In ordinary QFT, the representations of the permutation group arising in spacetime dimensions greater than two\footnote{Because of topological reasons, in two dimensions we may have a sort of "in-between" statistics, namely $\ket{p_1, p_2} = e^{i \theta} \ket{p_2, p_1}$, corresponding to the so-called \textit{anyons}.} are one-dimensional, as they just amount to a sign change: $\ket{\vec{p}_1, \vec{p}_2} = \ket{\vec{p}_1, \vec{p}_2}$ for bosons and $\ket{\vec{p}_1, \vec{p}_2} = - \ket{\vec{p}_1, \vec{p}_2}$ for fermions. When representations are higher-dimensional, we refer to them as \textit{parastatistics}. A \textit{parabosonic} field is defined as \cite{Parastatistics}

\begin{equation}
    \phi (x) = \sum_{i} \phi^{(i)} (x)
    \label{par_field}
\end{equation}
with the fields $\phi^{(i)}$ satisfying $[\phi^{(i)}(x), \phi^{(i)}(y)] = 0$ and $\{\phi^{(i)}(x), \phi^{(j)}(y) \} = 0$, for $i \neq j$ and $x$, $y$ spacelike-separated. A \textit{parafermionic} field is analogously defined, switching commutators and anticommutators. When acting on their respective Fock spaces, these fields give rise to higher-dimensional representations of the permutation group: because the field operators neither commute nor anticommute, $\phi(x_1) \phi(x_2) \ket{\Omega}$ and $\phi(x_2) \phi(x_1) \ket{\Omega}$ lie in different rays of the Hilbert space. The operator $\epsilon$ introduced above acts on Fock states as $\epsilon \phi(x_1) \phi(x_2) \ket{\Omega} = \phi(x_2) \phi(x_1) \ket{\Omega}$.

\subsection{The index for charged black holes}

We can take \ref{teleport} further and specify the value of $[\mathcal{M}:\mathcal{N}]$ in the case of charged black holes discussed above. From the index-statistics theorem, we know that

\begin{equation}
    [\mathcal{A}(\mathcal{O}):\rho \big( \mathcal{A}(\mathcal{O}) \big)] = \big( d(\rho) \big)^2
\end{equation}
with $d(\rho)$ being the statistical dimension of the superselection sector $\rho$.
Since we are discussing an evaporation process, it is natural to assume a charge loss formalized by a change in the superselection sector yielding $\rho' \big( \mathcal{A}(\mathcal{O}) \big) \subset \rho \big( \mathcal{A}(\mathcal{O}) \big)$. This is consistent with our previous notation, where $\mathcal{M}$ and $\mathcal{N}$ are respectively replaced by $\rho' \big( \mathcal{A}(\mathcal{O}) \big)$ and $\rho \big( \mathcal{A}(\mathcal{O}) \big)$.

\begin{Proposition}
    Assuming $\rho' \big( \mathcal{A}(\mathcal{O}) \big) \subset \rho \big( \mathcal{A}(\mathcal{O}) \big)$, the following relation holds:

    \begin{equation}
        [\rho \big( \mathcal{A}(\mathcal{O}) \big) : \rho' \big( \mathcal{A}(\mathcal{O}) \big)] = \bigg( \frac{d(\rho')}{d(\rho)} \bigg)^2
    \end{equation}
\end{Proposition}

\begin{proof}
    The proof follows trivially from the index-statistics theorem and the property $[\mathcal{M}:\mathcal{N}] = [\mathcal{M}:\mathcal{R}][\mathcal{R}:\mathcal{N}]$ for an inclusion $\mathcal{N} \subset \mathcal{R} \subset \mathcal{M}$.
\end{proof}

The information-theoretic nature of the statistical dimension, which directly follows from the index-statistics theorem, uncovers a link between charge loss and information loss which has already been examined in \cite{Longo2}. Essentially, we may think of it as follows: the black hole, losing charge, performs the switch $\rho \rightarrow \rho'$ with the latter inequivalent to the former. Since $[\mathcal{M}:\mathcal{N}] > 1$, we have $d(\rho') > d(\rho)$ meaning $\rho' \big( \mathcal{A} (\mathcal{O}) \big ) \subset \rho \big( \mathcal{A} (\mathcal{O}) \big )$. This means that the endomorphism describing the superselection sector undergoes an information loss due to being more restricted than before.

The formula \ref{teleport final} can thus be completed as

\begin{equation}
        \bra{\Psi} e_{\mathcal{M}_1} \Gamma (A) e_{\mathcal{M}_1} \ket{\Psi} = \bigg( \frac{d(\rho)}{d(\rho')} \bigg)^2 \bra{\Psi} A \ket{\Psi} \text{.}
    \label{teleport charge}
\end{equation}
Furthermore, the statistical dimension can be related to the relative free $\mathcal{F}$ energy of a black hole \cite{Longo2}, which is quantized as $d(\rho) \in \mathbb{N}$; we have

\begin{equation}
    \mathcal{F}(\Omega||\Psi_{\rho}) = - \frac{1}{2} \beta^{-1} \log \big( d(\rho) \big)
\end{equation}
yielding

\begin{equation}
    \bra{\Psi} e_{\mathcal{M}_1} \Gamma (A) e_{\mathcal{M}_1} \ket{\Psi} = \frac{e^{2 \beta \mathcal{F}(\Omega||\Psi_{\rho'}) }}{e^{2 \beta \mathcal{F}(\Omega||\Psi_{\rho})}} \bra{\Psi} A \ket{\Psi} \text{.}
\end{equation}
The Jones index cannot take arbitrary values: it must belong to $\{ 4 \cos^2(\frac{\pi}{n}) \ | \ n \in \mathbb{N} \} \cup [4, \infty]$. An important constraint that needs to be imposed in order for the overall discussion to have physical meaning is $\mathcal{A}, \ \mathcal{B} \neq \mathbb{C} I$, so that the information transferred is non-trivial. The simple request of having a non-trivial relative commutant, which physically corresponds to having some actual information being evaporated out, forces $[\mathcal{M} : \mathcal{N}] > 4$ \cite{Jones} \cite{Kosaki}. This argument does not rely on DHR theory, but is rather a purely information-theoretic argument: given a generic physical scenario modelled by an inclusion $\mathcal{N} \subset \mathcal{M}$ with finite index (such as $\mathcal{N} \otimes I \subset \mathcal{N} \otimes M_2(\mathbb{C})$ for a qubit coupled to a type $\mathrm{III}$ environment), if we impose $\mathcal{A} \neq \mathbb{C} I$ then we automatically have a minimum index value possible. When the inclusion $\mathcal{N} \subset \mathcal{M}$ is realized by a DHR inlcusion, if a non-trivial amount of information is emitted from a sector $\rho$, then due to the previous bound the possible final sector $\rho'$ is not arbitrary if the localization of the region stays about the same: the new endomorphism $\rho'$ must contain a minimum amount of information less then $\rho$. As a result, the final charge of the state is not arbitrary either.
From a field-theoretic point of view, if we have a $d$-dimensional representation of the parastatistical field $\phi$ as in \ref{par_field} describing the algebra $\mathcal{M}$, then the new algebra $\mathcal{N}$ consists of fields living in a $d'$-dimensional representation with $d' > 2d$. The more information is being evaporated through Hawking radiation, the more "elementary" fields $\phi^{(i)}$ we need to construct the parastatistical field $\phi$. Naturally, there is a further constraint in the case of charged DHR states, coming from the fact that $d(\rho) \in \mathbb{N} \cup \{ \infty \}$, which discretizes the values of the index. Perhaps, there may be mechanisms involving the appearance of the "elementary" fields $\phi^{(i)}$ related to quantum gravity corrections. This may for instance be the case in non-unitary models, where a non-unitary map may relate different superselection sectors during the evolution.

Alternatively, if the black hole horizon has a quantized area, as conjectured by some quantum gravity models, the emitted information will be discretized, possibly yielding a finite-dimensional relative commutant. In such case, the protocol described above would be fitting too, although there is no guarantee the algebras stay type $\mathrm{III}$ in that context.

\section{Conclusion}
In this work, we presented a natural generalization for the quantum information retrieval protocol for type $\mathrm{III}$ inclusions of factors, based on the work of Verlinde and van der Heijden. The mathematical framework adopted exploits the notion of Jones index for the inclusions of type $\mathrm{III}$ von Neumann algebras, developed by Kosaki and Longo.

The formula obtained finds a natural setting in the adiabatic evaporation of charged black holes, which undergo a change in their superselection sector due to charge loss.
The non-triviality of Alice's and Bob's algebras leads to a constraint on the minimal charge emitted during an evaporation step, providing a new vision on the quantization of charge from information-theoretic arguments.

Future developments may involve multiple directions. Concerning possible generalizations, the protocol could be extended to the case of von Neumann algebras with non-trivial centers. From a more physical point of view, another viable development could take dynamics into account, which was neglected in this treatment.

Other general ideas could concern the investigation of parastatistical models coupled to gravity, to analyse explicitly the interplay between gravity and the statistical dimension. Moreover, the role played by embezzlement in the context of LOCC maps in QFT could lead to severe differences with respect to the type $\mathrm{I}$ case. More generally, finding other applications of the Jones index for quantum information tasks in QFT could shed a new light in relativistic quantum information theory.

\acknowledgments
I want to thank Roberto Longo for the many fruitful discussions we had during the development of this work. His help has been essential to the writing of this paper. I also want to thank Massimo Taronna for sparking my interest on the role played by von Neumann algebras in frontier physics and inspiring me to work on the topic of black hole information.

I also acknowledge support from the INFN Iniziativa Specifica QUAGRAP and from the European COST Actions BridgeQG CA23130, RQI CA23115 and CaLISTA CA21109.

\newpage
\appendix
\section{Appendix}
\subsection{Generalities on von Neumann algebras}

In this section, we shall introduce the concept of von Neumann algebras and the role played by them in QFT \cite{Von Neumann} \cite{Naimark} \cite{Bratteli-Robinson} \cite{Sorce} \cite{Yvangson}.

\begin{Definition}
    Given a subset $\mathcal{A} \in \mathcal{B}(\mathcal{H})$, we define $\mathcal{A}'$ to be $\mathcal{A}$'s \textnormal{commutant} if and only if
    \begin{equation}
        \mathcal{A}' = \{ A \in \mathcal{B}(\mathcal{H}) \ | \ AB=BA, \ \forall B \in \mathcal{A} \} \text{.}
    \end{equation}
\end{Definition}

\begin{Definition}
    A subset $\mathcal{A} \in \mathcal{B}(\mathcal{H})$ is called \textnormal{von Neumann algebra} if and only if $\mathcal{A} = \mathcal{A}''$.
\end{Definition}

\begin{Definition}
   Given a subalgebra $\mathcal{A}$ of $\mathcal{B}(\mathcal{H})$, we define its \textnormal{center} to be the set $Z = \mathcal{A} \cap \mathcal{A}'$.
\end{Definition}

\begin{Definition}
    A \textnormal{factor} is a von Neumann algebra $\mathcal{A}$ whose center is $\mathbb{C}I$.
\end{Definition}

Factors can be classified basing on a quantity $d(P) \in \mathbb{R}_0^+ \cup \{ \infty \}$ defined on a projection $P \in \mathcal{A}$. The quantity $d$ has the suitable properties of a dimension and indeed, in familiar cases, it coincides with the dimension of the subspace onto which $P$ projects the vectors: for instance, a minimal projector in $\mathcal{B}(\mathcal{H})$ is any projector whose range is one-dimensional (such as $\ket{\psi}\bra{\psi}$). For such reason, it is called \textit{relative dimension}.

\begin{Proposition}[\textbf{Classification of von Neumann factors}]
    Given the set $\mathfrak{D}$ of the values the function $d$ can assume, we can only have the following cases:
    \begin{itemize}
        \item the factor $\mathcal{A}$ is \textnormal{type $\mathrm{I}_n$}: $\mathfrak{D} = \{ k \lambda | k \in \mathbb{N}_0, k \leq n \}, \lambda > 0$;
        \item the factor $\mathcal{A}$ is \textnormal{type $\mathrm{I}_{\infty}$}: $\mathfrak{D} = \{ n \lambda | n \in \mathbb{N}_0 \cup \{ \infty \} \}, \lambda > 0$;
        \item the factor $\mathcal{A}$ is \textnormal{type $\mathrm{II}_1$}: $\mathfrak{D} = [0, \lambda], \lambda > 0$;
        \item the factor $\mathcal{A}$ is \textnormal{type $\mathrm{II}_{\infty}$}: $\mathfrak{D} = [0, \infty]$;
        \item the factor $\mathcal{A}$ is \textnormal{type $\mathrm{III}$}: $\mathfrak{D} = \{0, \infty \}$.
    \end{itemize}
    \label{Classification of factors}
\end{Proposition}

If $\mathcal{A} = \mathcal{B}(\mathcal{H})$, we are always in the case of a factor of type $\mathrm{I}_D$ where $D = \mathrm{dim}(\mathcal{H})$. That is the case in quantum mechanics (even when we consider subsystems, given the tensor product factorization). A stronger statement actually holds:

\begin{Proposition}
    A factor $\mathcal{A}$ is type $\mathrm{I}$ if and only if it is isomorphic to $\mathcal{B}(\mathcal{H})$.
    \label{Factors}
\end{Proposition}

Let us briefly discuss now the relevance of this classification in QFT. As an illustrative example, consider the right Rindler wedge $\mathcal{R}$ of Minkowski spacetime. In this case, the algebra $\mathcal{A}(\mathcal{R})$ is a type $\mathrm{III}$ von Neumann factor and the algebra $\mathcal{A}(\mathcal{L})$ associated with the left Rindler wedge is its commutant \cite{Araki}.

A thing worth emphasising is the natural compatibility between local quantum physics and observables algebras being type $\mathrm{III}$. In fact, \textit{any projection in a type $\mathrm{III}$ factor is equivalent to the identity}; this is trivial since the relative dimension $d$ is always infinite for any non-zero projection, thus we always have $d(P) = d(Q)$ for any $P, \ Q \neq 0$, implying $P \sim Q$, namely there is a partial isometry $V$ such that $P = VV^*$ and $Q = V^* V$. The implication of this property in QFT is the following: we can perform local operations that change any state $\omega$ into an eigenstate of a local projection $P \in \mathcal{A}(\mathcal{U})$, without disturbing it in its causal complement. As a matter of fact, since there is an isometry $V \in \mathcal{A}(\mathcal{U})$ such that $P = V V^*$ and $I = V^* V$,

\begin{equation}
    \omega_V(P) = \omega(V^* P V) = \omega(V^* V V^* V) = \omega(I) = 1
\end{equation}
but

\begin{equation}
    \omega_V(B) = \omega(V^* B V) = \omega (V^* V B) = \omega(B), \ \forall B \in \mathcal{A}(\mathcal{U}') \subseteq \big( \mathcal{A}(\mathcal{U}) \big)' \text{,}
\end{equation}
with $\omega_V(\cdot) \eqdef \omega \big( V^* (\cdot) V \big)$.

We can interpret the classification discussed above as follows. Let us consider a quantum system described by a Hilbert space $\mathcal{H}$, whose corresponding algebra of observables is $\mathcal{B}(\mathcal{H})$. We may ask \textit{whether it is possible to construct density matrices from an algebra of observables}. That is certainly the case for $\mathcal{B}(\mathcal{H})$, as any orthonormal sequence $\{ \ket{\psi_n} \}$ with $ \{p_n\} \subseteq [0, 1]$ such that $\sum_n p_n = 1 $ defines the density matrix $\sum_n p_n \ket{\psi_n} \bra{\psi_n}$. Moreover, a state vector $\ket{\psi}$ in $\mathcal{H}$ represents a pure state, as represents the density matrix $\rho_{\psi} = \ket{\psi}\bra{\psi}$ associated with it. The operator $\rho_{\psi}$ is a \textit{minimal projection}, in the sense that there is no other non-trivial projection $P < \rho_{\psi}$; in fact, $d(\rho_{\psi}) = 1$ after a suitable normalization and $\mathfrak{D}$ is discrete. We can associate the presence of minimal projections to the existence of pure states. Thus, type $\mathrm{I}$ factors admit pure states within them.

In a type $\mathrm{II}$ factor, there are no minimal projections, because $\mathfrak{D}$ is continuous; as such, no pure state is constructable. However, we might think of constructing renormalizable density matrices: since finite\footnote{A projection $P$ is \textit{finite} iff $Q \leq P, \ Q \sim P \implies Q=P$.} projections exist, we can define a renormalized trace $\mathrm{Tr}_{ren}$ which is finite on them and infinite on infinite projections\footnote{Essentially, $\mathrm{Tr}_{ren}$ is realized by the relative dimension $d$.}. Indeed, if we want to preserve all the nice properties of the usual trace, we must have an infinite trace value for an infinite projection $Q$, i.e. such that there is a non-zero $P < Q$ satisfying $P \sim Q$. As we want $\mathrm{Tr}_{ren}$ to be the same on equivalent projections (in order to have $\mathrm{Tr}_{ren} (V^* V) = \mathrm{Tr}_{ren} (V V^*)$),

\begin{equation}
    \mathrm{Tr}_{ren} (Q) = \mathrm{Tr}_{ren} (P) + \mathrm{Tr}_{ren} (Q - P) = \mathrm{Tr}_{ren} (Q) + \mathrm{Tr}_{ren}(Q - P) \text{,}
\end{equation}
from which it follows that $\mathrm{Tr}_{ren}(Q - P) = 0 \implies Q - P = 0$, which is a contradiction. Thus, $\mathrm{Tr}_{ren} (Q) = \infty$. This shows that a renormalized trace with the desirable properties of a trace can only be defined for finite projections.

The type $\mathrm{III}$ case is the most peculiar, as there are neither minimal nor finite projections. In that case, every non-zero projection is equivalent to the identity.

To summarize,

\begin{center}
    $\mathcal{A}$ is type $\mathrm{I}$ $\longleftrightarrow$ there are \textnormal{pure} (and mixed) \textnormal{states} \\
    $\mathcal{A}$ is type $\mathrm{II}$ $\longleftrightarrow$ there are (renormalizable) \textnormal{mixed states} \\
    $\mathcal{A}$ is type $\mathrm{III}$ $\longleftrightarrow$ there are \textnormal{no renormalizable states}
\end{center}
A core result in the theory of operator algebras is Tomita-Takesaki theorem.

\begin{Proposition}[\textbf{Tomita-Takesaki theorem}]
    Let $\mathcal{A}$ be a von Neumann algebra with a cyclic and separating vector $\ket{\Omega}$. Calling $S_{\Omega}$ the minimal closed extension of the application $A \ket{\Omega} \rightarrow A^* \ket{\Omega}$, its polar decomposition is
    
    \begin{equation}
        S_{\Omega} = J_{\Omega} \Delta_{\Omega}^{\frac{1}{2}} \text{,}
    \end{equation}
    where $J_{\Omega}$, called \textnormal{modular conjugation}, is an anti-isometry and $\Delta_{\Omega} = S_{\Omega}^* S_{\Omega}$ is a positive self-adjoint operator, called \textnormal{modular operator}. Tomita-Takesaki theorem states that the relations
    
    \begin{equation}
        \begin{cases}
            &\Delta_{\Omega}^{-it} \mathcal{A} \Delta_{\Omega}^{it} = \mathcal{A} \\
            &J_{\Omega} \mathcal{A} J_{\Omega} = \mathcal{A}'
        \end{cases}
    \end{equation}
    hold. The set of automorphisms $\{ \sigma_t^{\Omega} (\cdot) \}_{t \in \mathbb{R}} \eqdef \{ \Delta_{\Omega}^{-it} (\cdot) \Delta_{\Omega}^{it} \}_{t \in \mathbb{R}}$ is called \textnormal{modular automorphism group}.
    \label{Tomita-Takesaki}
\end{Proposition}

\begin{Proposition}[\textbf{KMS condition}]
    Let $\mathcal{A}$ be a von Neumann algebra with cyclic and separating vector $\ket{\Omega}$ and let $\omega$ be the state such that $\omega(A) = \bra{\Omega} A \ket{\Omega}$, for any $A \in \mathcal{A}$. The state $\omega$ satisfies the KMS condition with respect to its modular automorphisms group $\{ \sigma_t^{\omega} \}_{t \in \mathbb{R}}$.
    
    Conversely, suppose $\mathcal{A}$ is a $C^*$-subalgebra of $\mathcal{B}(\mathcal{H})$ with a cyclic vector $\ket{\Omega}$. If $\omega$ satisfies the KMS condition with respect to an automorphisms group $\{ \tilde{\sigma}_t \}_{t \in \mathbb{R}}$, then:

    \begin{itemize}
        \item $\omega \big( \tilde{\sigma}_t(A) \big) = \omega(A), \ \forall t \in \mathbb{R}, \ \forall A \in \mathcal{A}$;
        \item $\ket{\Omega}$ is separating for the von Neumann algebra $\mathcal{M} \eqdef \mathcal{A}''$;
        \item $\tilde{\sigma}_t$ is the restriction to $\mathcal{A}$ of the modular automorphism $\sigma_t^{\Omega}$ of $\mathcal{M}$ defined by $\ket{\Omega}$, for every $t \in \mathbb{R}$.
    \end{itemize}
\end{Proposition}

This theorems state that, given a fully entangled\footnote{A cyclic and separating vector corresponds to a state in the form $\sum_n c_n \ket{n} \otimes \ket{n} \in \mathcal{H} \otimes \mathcal{H}$, with $c_n \neq 0, \ \forall n$ \cite{Witten}.} state, we can always find a suitable Hamiltonian such that, for the observer telling time according to such Hamiltonian, the state looks thermal. This holds for any von Neumann algebras and, as a result, applies also to QFT. The operator $- \log(\Delta_{\Omega})$ is called \textit{modular Hamiltonian}, since it satisfies the KMS condition \cite{Kubo} \cite{Martin-Schwinger}

\begin{equation}
    \bra{\Omega} \sigma_t^{\Omega}(A) B \ket{\Omega} = \bra{\Omega} B \sigma_{t+i}^{\Omega}(A) \ket{\Omega}, \ \forall A, B \in \mathcal{A}
\end{equation}
at inverse temperature $\beta = 1$. The Unruh effect can be deduced in this framework by deriving the modular Hamiltonian corresponding to a Rindler wedge \cite{Bisognano-Wichmann} \cite{Unruh}.

\subsection{Jones index}

In this section, we will review the basics of Jones index theory \cite{Jones}, following the synthesis provided in \cite{Verlinde}. Consider a von Neumann algebra $\mathcal{M}$. Beyond its standard GNS representation in the Hilbert space $\mathcal{H}_{\mathcal{M}}$, one can construct other representations too. An example would be to consider the Hilbert space

\begin{equation}
    \mathcal{H}^{(n)} \eqdef \mathcal{H}_{\mathcal{M}} \oplus ... \oplus \mathcal{H}_{\mathcal{M}}
\end{equation}
which consists of $n$ copies of the GNS representation $\mathcal{H}_{\mathcal{M}}$. In such a Hilbert space, the vector $\ket{\Psi} \oplus ... \oplus \ket{\Psi}$ is no longer cyclic, as it only generates the diagonal elements of $\mathcal{H}^{(n)}$. An infinite generalization can also be considered, by taking the space

\begin{equation}
    \mathcal{H}^{(\infty)} \eqdef \ell^2(\mathbb{N}) \otimes \mathcal{H}_{\mathcal{M}} \text{,}
\end{equation}
that is the $\mathcal{H}_{\mathcal{M}}$-valued (square-integrable) sequences \footnote{Analogously to Lie algebra-valued one-forms in gauge theories. In fact, we can associate to each $\bigoplus_{n=1}^{\infty} \ket{\psi_n}$ the sequence $\{ c_k^{(n)} \}$ to be contracted with the basis $\{ \ket{n} \otimes \ket{\varphi_k} \}$, with $\ket{\psi_n} = \sum_{k=1}^{\infty} c_k^{(n)} \ket{\varphi_k}$, and vice versa.}. In general, one may think of constructing representations $\alpha$ times as big as the GNS one, with $\alpha \in [0, \infty]$. Such $\alpha$ will be the $\mathcal{M}$-dimension of the representation and quantifies the relative size of a representation $\mathcal{H}$ with respect to the standard one $\mathcal{H}_{\mathcal{M}}$. We require that this quantity should be equal to $1$ when $\mathcal{H} = \mathcal{H}_{\mathcal{M}}$ and should be the same for $\mathcal{H}'$ unitarily equivalent to $\mathcal{H}$. The following theorem allows us to define this dimension.

\begin{Proposition}
    Let $\mathcal{M}$ be a type $\mathrm{II}_1$ factor and $\mathcal{H}$ a representation of $\mathcal{M}$ \footnote{To be precise, $\mathcal{H}$ should be an $\mathcal{M}$-module.}. Then, there is an isometry

    \begin{equation}
        V : \mathcal{H} \rightarrow \mathcal{H}^{(\infty)}
    \end{equation}
    such that

    \begin{equation}
        V M = (I \otimes M) V, \ \forall M \in \mathcal{M} \text{.}
    \end{equation}
    Additionally,

    \begin{equation}
        V V^* \in (I \otimes \mathcal{M})' = \mathcal{B} \big( \ell^2(\mathbb{N}) \big) \otimes \mathcal{M}'
    \end{equation}
    defines a projection and, given the trace $\tau \eqdef \mathrm{tr} \otimes \mathrm{Tr}$ on the commuting algebra, $\tau (VV^*)$ is independent of $V$.
\end{Proposition}

\begin{Definition}
    The $\mathcal{M}$-\textnormal{dimension} of $\mathcal{H}$ it is defined as

    \begin{equation}
        \mathrm{dim}_{\mathcal{M}} (\mathcal{H}) \eqdef \tau(VV^*) \in [0, \infty] \text{.}
    \end{equation}
\end{Definition}

This definition satisfies the following properties:

\begin{enumerate}
    \item $\mathrm{dim}_{\mathcal{M}} (\mathcal{H}) = \mathrm{dim}_{\mathcal{M}} (\mathcal{H}') \iff \mathcal{H} \cong \mathcal{H}'$;
    \item if $\mathrm{dim}_{\mathcal{M}} (\mathcal{H}) = \infty$ then $\mathcal{M}'$ is type $\mathrm{II}_{\infty}$, otherwise it is type $\mathrm{II}_1$;
    \item for a countable set $\{ \mathcal{H}_k \}$, 
    
    \begin{equation}
        \mathrm{dim}_{\mathcal{M}} \Big( \bigoplus_k \mathcal{H}_k \Big) = \sum_k \mathrm{dim}_{\mathcal{M}} (\mathcal{H}_k) \text{;}
    \end{equation}
    \item for a projection $P \in \mathcal{M}'$,
    
    \begin{equation}
        \mathrm{dim}_{\mathcal{M}} (\mathcal{H} P) = \mathrm{Tr}_{\mathcal{M}'}(P) \ \mathrm{dim}_{\mathcal{M}} (\mathcal{H}) \text{;}
    \end{equation}
    \item for a projection $P \in \mathcal{M}$,
    
    \begin{equation}
        \mathrm{dim}_{P \mathcal{M} P} (P \mathcal{H}) = \frac{1}{\mathrm{Tr}_{\mathcal{M}}(P)} \mathrm{dim}_{\mathcal{M}} (\mathcal{H}) \text{;}
    \end{equation}
    \item $\mathrm{dim}_{\mathcal{M}'}(\mathcal{H}) = \frac{1}{\mathrm{dim}_{\mathcal{M}}}(\mathcal{H})$.
\end{enumerate}

\begin{Definition}
    Let $\mathcal{N} \subset \mathcal{M}$ be an inclusion of type $\mathrm{II}_1$ factors. The \textnormal{Jones index} is defined as

    \begin{equation}
        [\mathcal{M} : \mathcal{N}] \eqdef \mathrm{dim}_{\mathcal{N}}(\mathcal{H}_{\mathcal{M}}) \text{.}
    \end{equation}
\end{Definition}

Since $\mathcal{H}_{\mathcal{N}} \subset \mathcal{H}_{\mathcal{M}}$, we can write

\begin{equation}
    \begin{split}
        &\mathcal{H}_{\mathcal{M}} = \mathcal{H}_{\mathcal{N}} \oplus \mathcal{H}_{\mathcal{N}}^{\perp} \implies \mathrm{dim}_{\mathcal{N}}(\mathcal{H}_{\mathcal{M}}) = \mathrm{dim}_{\mathcal{N}}(\mathcal{H}_{\mathcal{N}}) + \mathrm{dim}_{\mathcal{N}}(\mathcal{H}_{\mathcal{N}}^{\perp}) = \\
        &1 + \mathrm{dim}_{\mathcal{N}}(\mathcal{H}_{\mathcal{N}}^{\perp}) \implies [\mathcal{M} : \mathcal{N}] \geq 1 \text{.}
    \end{split}
\end{equation}

\begin{Proposition}
    Given $\mathcal{N} \subset \mathcal{M}$ with a representation $\mathcal{H}$ such that $\mathrm{dim}_{\mathcal{N}}(\mathcal{H}) < \infty$,

    \begin{equation}
        [\mathcal{M} : \mathcal{N}] = \frac{\mathrm{dim}_{\mathcal{N}}(\mathcal{H})}{\mathrm{dim}_{\mathcal{M}}(\mathcal{H})} \text{.}
    \end{equation}
\end{Proposition}

\begin{Proposition}
    If $[\mathcal{M} : \mathcal{N}] < \infty$,
    \begin{equation}
        [\mathcal{M}_1 : \mathcal{M}] = [\mathcal{M} : \mathcal{N}] \text{.}
    \end{equation}
\end{Proposition}

We can now prove the relation \ref{Jind} used in the main text.

\begin{Proposition}
    \begin{equation}
        \mathrm{Tr}_{\mathcal{M}_1}(e_{\mathcal{N}} M) = [\mathcal{M} : \mathcal{N}]^{-1} \mathrm{Tr}_{\mathcal{M}}(M), \ \forall M \in \mathcal{M} \text{.}
    \end{equation}
\end{Proposition}

\begin{proof}
    To prove the statement, it is useful to first prove the relation

    \begin{equation}
        \mathrm{Tr}_{\mathcal{M}_1}(e_{\mathcal{N}}) = [\mathcal{M} : \mathcal{N}]^{-1} \text{.}
    \end{equation}
    Indeed,

    \begin{equation}
        \begin{split}
            1 &= \mathrm{dim}_{\mathcal{N}}(\mathcal{H}_{\mathcal{N}}) = \mathrm{dim}_{\mathcal{N}}(e_{\mathcal{N}} \mathcal{H}_{\mathcal{M}}) = \mathrm{Tr}_{\mathcal{N}'}(e_{\mathcal{N}}) \ \mathrm{dim}_{\mathcal{N}}(\mathcal{H}_{\mathcal{M}}) \\
            &= \mathrm{Tr}_{\mathcal{N}'}(e_{\mathcal{N}}) \ [\mathcal{M} : \mathcal{N}] \text{.}
        \end{split}
    \end{equation}
    Consequently, since $\mathcal{M}_1 = J_{\mathcal{M}} \mathcal{N}' J_{\mathcal{M}}$ and $e_{\mathcal{N}} = J_{\mathcal{M}} e_{\mathcal{N}} J_{\mathcal{M}}$,

    \begin{equation}
        \mathrm{Tr}_{\mathcal{M}_1}(e_{\mathcal{N}}) = \mathrm{Tr}_{J_{\mathcal{M}} \mathcal{N}' J_{\mathcal{M}}} (J_{\mathcal{M}} e_{\mathcal{N}} J_{\mathcal{M}}) = \mathrm{Tr}_{\mathcal{N}'} (e_{\mathcal{N}}) = [\mathcal{M} : \mathcal{N}]^{-1} \text{.}
    \end{equation}
    Considering $N_1, N_2 \in \mathcal{N}$,

    \begin{equation}
        \mathrm{Tr}_{\mathcal{M}_1}(N_1 N_2 e_{\mathcal{N}}) = \mathrm{Tr}_{\mathcal{M}_1}(N_1 e_{\mathcal{N}} N_2) = \mathrm{Tr}_{\mathcal{M}_1}(N_2 N_1 e_{\mathcal{N}})
    \end{equation}
    as $e_{\mathcal{N}} \in \mathcal{N}'$. Consequently, the map

    \begin{equation}
        N \in \mathcal{N} \rightarrow \mathrm{Tr}_{\mathcal{M}_1}(N e_{\mathcal{N}})
    \end{equation}
    satisfies the tracial property on $\mathcal{N}$. As a result, there is a constant $\alpha$ such that

    \begin{equation}
        \mathrm{Tr}_{\mathcal{M}_1} (N e_{\mathcal{N}}) = \alpha  \ \mathrm{Tr}_{\mathcal{N}}(N) = \alpha  \ \mathrm{Tr}_{\mathcal{M}}(N), \ \forall N \in \mathcal{N} \text{,}
    \end{equation}
    where we used the fact that $\mathrm{Tr}_{\mathcal{N}} = \left. \mathrm{Tr}_{\mathcal{M}} \right|_{\mathcal{N}}$. Evaluating the above expression for $N = I$,

    \begin{equation}
        \alpha = \mathrm{Tr}_{\mathcal{M}_1} (e_{\mathcal{N}}) = [\mathcal{M} : \mathcal{N}]^{-1}
    \end{equation}
    and thus, for $M \in \mathcal{M}$,

    \begin{equation}
        \begin{split}
            \mathrm{Tr}_{\mathcal{M}_1} (M e_{\mathcal{N}}) &= \mathrm{Tr}_{\mathcal{M}_1} (M e_{\mathcal{N}}^2) = \mathrm{Tr}_{\mathcal{M}_1} (e_{\mathcal{N}} M e_{\mathcal{N}}) \\
            &= \mathrm{Tr}_{\mathcal{M}_1} \big( \mathcal{E}(M) e_{\mathcal{N}} \big) = \alpha \ \mathrm{Tr}_{\mathcal{M}} \big( \mathcal{E}(M) \big) = \alpha \ \mathrm{Tr}_{\mathcal{M}}(M)
        \end{split}
    \end{equation}
    using $e_{\mathcal{N}} M e_{\mathcal{N}} = \mathcal{E}(M) e_{\mathcal{N}}$ and $\mathrm{Tr}_{\mathcal{M}} \circ \mathcal{E} = \mathrm{Tr}_{\mathcal{M}}$.
\end{proof}
\appendix




\end{document}